\documentclass[10pt,conference,letterpaper]{IEEEtran}
\usepackage{times,amsmath,epsfig,url,graphics,subfig,balance,mathtools,amsmath,amssymb,multirow,xcolor}
\usepackage[linesnumbered,ruled,algonl,vlined,noend]{algorithm2e}

\newcommand{\var}[1]{\mbox{\emph{#1}}}
\newcommand{\avar}[1]{\mbox{{#1}}}

\newcommand{\ssvar}[1]{\mbox{\tiny\emph{#1}}}
\newcommand{\assvar}[1]{\mbox{\tiny{#1}}}
\newtheorem{lemma}{Lemma}

\newtheorem{exmp}{Example} 
\newtheorem{definition}{Definition}
\newcommand{\comm}[1]{}
\newcommand{\myparagraph}[1]{\vspace{0.0pt}\noindent{\textbf{#1.}}~}
\newcommand{\km}{{$k$MaxRRST}\xspace}
\newcommand{\tqtree}{{TQ-tree}\xspace}
\newcommand{\tqb}{{TQ(B)}\xspace}
\newcommand{\tqz}{{TQ(Z)}\xspace}
\newcommand{\bl}{{BL}\xspace}

\title{The Maximum Trajectory Coverage Query in Spatial Databases}
\author{\IEEEauthorblockN{Mohammed Eunus Ali$^1$ \qquad Kaysar Abdullah$^1$ \qquad Shadman Saqib Eusuf$^1$}
\IEEEauthorblockA{$^1$Bangladesh University of Engineering and Technology, Bangladesh\\
Email: eunus@cse.buet.ac.bd, kzr.buet08@gmail.com, s.saqibeusuf@gmail.com}\\

\IEEEauthorblockN{Farhana M. Choudhury$^2$ \qquad J. Shane Culpepper$^2$ \qquad Timos Sellis$^3$}
\IEEEauthorblockA{$^2$RMIT University, Australia, $^3$Swinburne University of Technology, Australia\\
Email: farhana.choudhury, shane.culpepper@rmit.edu.au, tsellis@swin.edu.au}}

\begin{document}
\maketitle
\begin{abstract}

With the widespread use of GPS-enabled mobile devices, an
unprecedented amount of trajectory data is becoming available from
various sources such as Bikely, GPS-wayPoints, and Uber.
The rise of innovative transportation services and recent
break-throughs in autonomous vehicles will lead to the
continued growth of trajectory data and related applications.
Supporting these services in emerging platforms will require
more efficient query processing in trajectory databases.
In this paper, we propose two new coverage queries for
trajectory databases: (i) {\em $k$ Maximizing Reverse Range Search on
Trajectories} ($k$MaxRRST);
and (ii) a {\em Maximum $k$ Coverage Range Search on Trajectories}
(Max$k$CovRST).
We propose a novel index structure, the Trajectory Quadtree
($TQ$-tree) that utilizes a quadtree to hierarchically organize
trajectories into different quadtree nodes, and then applies a
z-ordering to further organize the trajectories by spatial locality
inside each node.
This structure is highly effective in pruning the trajectory search
space, which is of independent interest.
By exploiting the $TQ$-tree data structure, we develop a divide-and-conquer approach to compute the trajectory ``service value'', and a best-first
strategy to explore the trajectories using the appropriate upper bound on the service value to efficiently process a $k$MaxRRST query.
Moreover, to solve the Max$k$CovRST, which is a \emph{non-submodular
NP-hard problem}, we propose a greedy approximation which also 
exploits the TQ-tree.
We evaluate our algorithms through an extensive experimental study on
several real datasets, and demonstrate that our
$TQ$-tree based algorithms outperform common baselines by two to
three orders of magnitude.
 
\end{abstract}
\section{Introduction}\label{intro}

With the widespread use of GPS-equipped mobile devices and the
popular map services, an unprecedented
amount of trajectory data is becoming available.
For example, in Bikely\footnote{\url{http://www.bikely.com}} users
can share their cycling routes from the GPS devices, in
GPS-wayPoints\footnote{\url{http://gpswaypoints.net}} a user can add
waypoints (points on a route at which a course is changed)
in a route and share with friends, in Microsoft
GeoLife\footnote{\url{https://research.microsoft.com/en-us/projects/geolife/}}
users can share their travel routes and experience using GPS
trajectories.
Most of the popular social network sites also support sharing user
trajectories

\begin{figure}
    \centering
        \includegraphics[width=2.8in, height=1.5in]{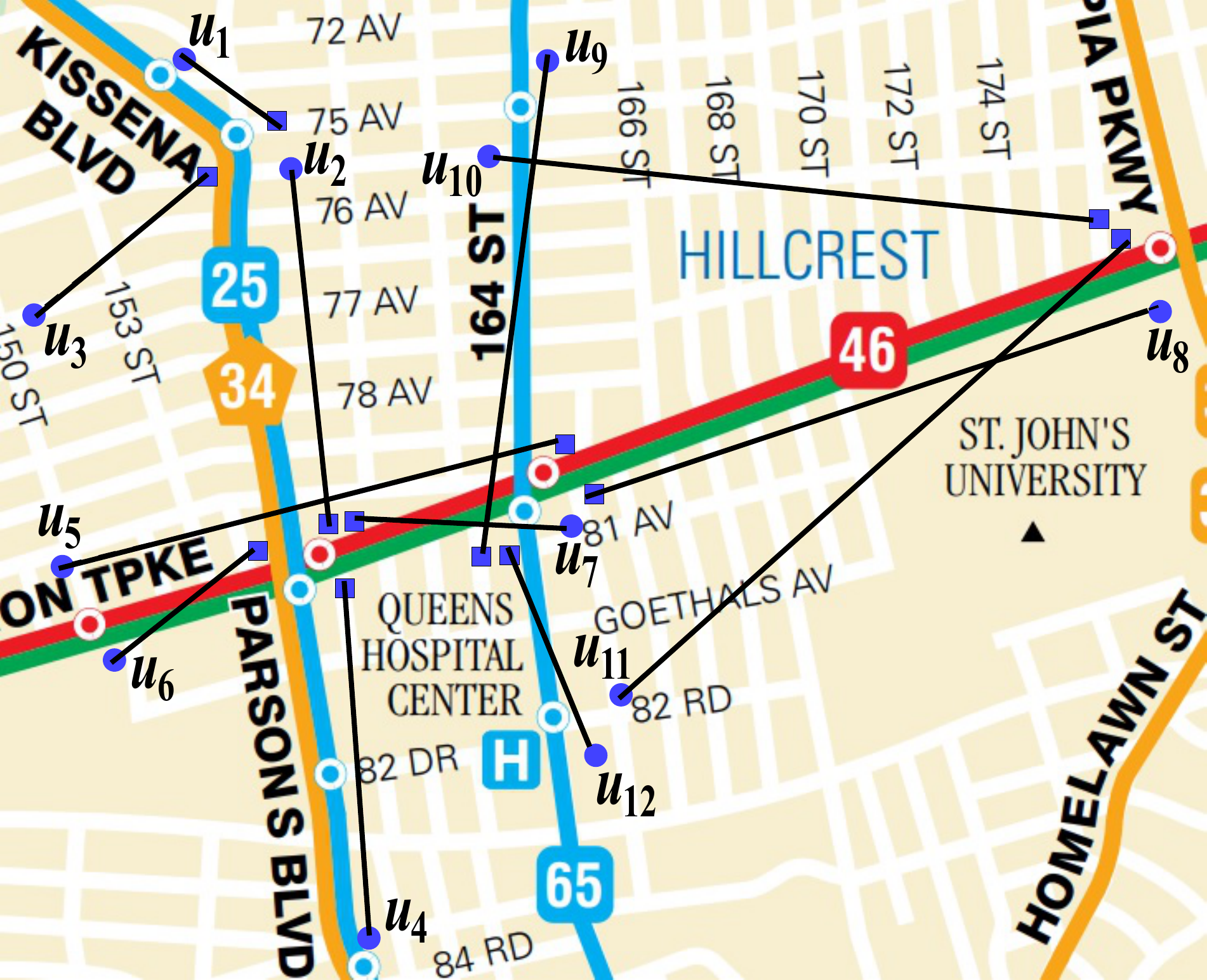}
        \vspace{-6pt}
    \caption{An example of a MaxRRST query and a Max$k$CovRST query
    with 12 user trajectories and 3 bus routes in NY, USA}
    \label{fig:example}
     \vspace{-16pt}
\end{figure}

Other than personal trips and travel routes, there are many examples
of trajectories from different transport services.
Uber served nearly $14.3$ million users in New York City between 
January-June 2015\footnote{\url{https://github.com/fivethirtyeight/uber-tlc-foil-response}}.
While user trajectory data has already been used for public transport
planning, a wide range of applications remain where planning ad-hoc
transport services are of interest.
As discussed in an IEEE Spectrum report earlier this year,
ride-sharing, taxi services, and on-demand transportation services
will be key sectors in the looming autonomous vehicle industry~{\cite{autocar}}.
Consider the following examples that highlight the potential
applications in planning ad-hoc transport services.

\emph{Scenario 1: }An autonomous transport service company wants to introduce new service routes that can serve the maximum number of users who are currently using other forms of transportation (e.g., personal
cars) for daily commute. The daily commuting routes are essentially a trajectory from a source to a destination.
Since there can be many possible service routes, the
top-$k$ routes that can serve the maximum number of users are of
interest.

\emph{Scenario 2: }Consider a tourist city, where
each tourist has a list of POIs to visit, which can
be seen as a trajectory.
Now, a tour operator wants to run a bus service in $k$ different
routes to serve the maximum number of tourists.
In this case, a tourist may use the service if a number of her POIs can be visited (if it is not possible to serve the full
list by the operator).

\emph{Scenario 3: }In a smart city, consider a public transport
operator wants to provide Wi-Fi service or display moving
advertisements to commuters.
The operator has multiple buses, and among the many possible routes,
the operator can choose the top-$k$ routes which provide this
additional service to the maximum number of users for the maximum duration.

The underlying problem in all the above scenarios is to select a
limited number (top-$k$) of facility trajectories from a given set that can best ``serve'' the user trajectories.
In Scenario 1, only the start and the end
points of each user trajectory are of interest, and a user can be served by a facility (e.g., ride the bus) if a stop of that facility is sufficiently close, i.e.,
within a certain distance $\psi$ to these locations. In this case, the service of a facility is a binary notion. 

In Scenarios 2 and 3, all of the points in a user trajectory can be important as one may want to maximize
the ``service'' to the user trajectories by a facility trajectory in
terms of the number of points (e.g., the
number of POIs that a tourist can
visit) or the trajectory length (e.g., the length of a journey with advertisement display). In this case, a user can be served partially by a facility. We use the term ``service'' of a facility to refer to both these binary and non-binary measures (details in Section~\ref{problem}). 

In this paper, we address this new class of trajectory search
problems, denoted as the {\bf $k$ Maximizing Reverse Range Search on
Trajectories ($k$MaxRRST)} query which finds $k$ facilities that maximize a service measure for a set of user trajectories.
Formally, given a set $U$ of user trajectories, a set $F$ of candidate
facility trajectories, and a positive integer $k$, a $k$MaxRRST query returns $k$ facilities from $F$ with the highest service to the user trajectories in $U$. We also address another variant of the query that returns $k$ facilitates from $F$ that \emph{combinedly} serve the maximum user trajectories from $U$.
We denote this type of query as a {\bf Maximum $k$ Coverage Range
Search on Trajectories (Max$k$CovRST)}.

As the service value, i.e., how well the users are
served by a service may vary across applications, we formally define
the service value function $\var{SO}(U,f)$ to measure the service of a facility $f$ on the set $U$ of user trajectories. 
For the Max$k$CovRST problem, the service value $\var{SO}(U,F^{\prime})$ is computed for a subset of facilities $F^{\prime} \subseteq F$, where the common service provided by different
facilities to the users are considered.
Please refer to Section~\ref{problem} for details.

\begin{exmp}
{\textit {Figure~\ref{fig:example} shows an example of a MaxRRST query for user trajectories $\{u_{1}, u_2, ..., u_{12}\}$, representing
daily routes of commuters, and
three facility trajectories $\{25, 46, 65\}$ representing
the bus routes with stop points in Queens NY. Let, a user will use a facility if there is a pickup/drop-off location
of that facility within a threshold distance from her source and
destination.
Thus, $u_1, u_2, u_4$ can be served by $25$, $u_5, u_6, u_7, u_8$ by $46$, and $u_9, u_{12}$ by $65$.
Hence the bus route $46$ will be returned as the answer for the query.
For Max$k$CovRST query, since $u_{10}, u_{11}$ can be served jointly by $46$ and
$65$, the answer of the query for $k=2$ is 
\{$46, 65$\} as they can serve maximum
$8$ users, $\{u_5$, $u_6$,
$u_7$, $u_8$, $u_9$, $u_{10}$, $u_{11}$, $u_{12}\}$, where the other sets of size $2$, $\{25, 46\}$ and $\{25, 65\}$ can serve $\{u_1$, $u_2$,
$u_4$, $u_5$, $u_6$, $u_7$, $u_8\}$, and $\{u_1$, $u_2$,
$u_4$, $u_9, u_{12}\}$, respectively.}}
\end{exmp}

A major challenge of these queries is to track the different segments
of a trajectory that can be served by a facility trajectory. A user can be served partially by a facility, and a user can be served by multiple facilities. 
In most of the existing work on trajectories (\cite{TangZXYYH11,
HanCZLW14}), the points of the trajectories are indexed using a
state-of-the-art spatial indexing method to answer a query.
However, such techniques are not amenable to our problem as {\bf both
the partial service of a user trajectory (e.g., the number of POIs
from the list of interesting places that a tourist can visit), and
the combined service of multiple facilities are required in this
problem}.
Similarly, previous studies that find trajectories within a
range of a query trajectory (\cite{ShangDZJKZ14}), or find the
reverse $k$ nearest neighbor trajectories (\cite{WangBCSC17}) cannot
be used for our problem, as it would require repeating the
approaches for each facility route, which is not efficient.
Moreover, to the best of our knowledge, there is no existing work on
trajectories that can be used to efficiently answer 
Max$k$CovRST, where a user trajectory can be served
jointly by multiple facility trajectories (See
Section~\ref{sec:related} for details).

The contributions of the paper are summarized as follows:

\begin{itemize}

\item We propose a new class of trajectory queries: (i) $k$ Maximizing Reverse Range Search on Trajectories
($k$MaxRRST) which finds $k$ query trajectories with the maximum service to the users.
(ii) Maximum $k$ Coverage Range Search on Trajectories (Max$k$CovRST)
which finds $k$ trajectories that combinedly maximize the service.

\item The novelty of our work comes from the key observation that if the points of multiple user trajectories are co-located and have similar orientation, then those trajectories are likely to be served by the same facility. We propose a novel two-level index structure, the Trajectory Quadtree
($TQ$-tree) based on this idea where such trajectories are stored together. Specifically, a quadtree structure is employed to organize the trajectories in a hierarchy, and then a z-ordering is applied to organize the trajectories
by spatial locality inside a quadtree node. {\textbf{ \emph Such a structure is highly effective in pruning the search space for different segments of trajectories based on locality and orientation, which is of independent interest.}}

\item We present an efficient divide-and-conquer approach where a facility trajectory is recursively divided and the service value of the components of the facility is calculated in that subspace. For each subspace, we apply a two-phase pruning technique using the {\tqtree}. As either the partial or the complete service values are important
based on the application, we present a best-first strategy to efficiently explore the facilities using the
appropriate upper bound on the service value. We also present different conditions where the process can be safely early terminated.

\item We prove that the Max$k$CovRST is a non-submodular NP-hard problem.
We propose an efficient two-step greedy
approximation algorithm to answer Max$k$CovRST, where in the first step we compute a subset of the highest serving facilities using our $k$MaxRRST algorithm, and then greedily choose $k$ facilities.

\item We evaluate
our algorithms through an extensive experimental study on real datasets.
The results demonstrate both the efficiency and effectiveness of the algorithms.

\end{itemize}

\section{Problem Formulation}
\label{problem}

Let $U$ be a set of user trajectories where each $u\in U$ is a
sequence of point locations, $u=\{p_1, p_2, ..., p_{|u|}\}$ and $F$ be a set of facility trajectories, where each $f\in F$ is a
sequence of stop points representing the pick-up or drop-off locations of a facility route (e.g., bus route).
A user trajectory can be served by a facility in different
contexts.

First, we present the calculation of the service values of a facility for a single user in different
scenarios, and then we present a generalized function to compute the
service value of a facility or a set of facilities for the set of users $U$.

\subsection{Service value for a single user}
\myparagraph{Scenario 1} Here, $u.p_1$ and $u.p_{|u|}$ are the source
and destination locations of $u$. 
A user $u$ is only interested in using a facility $f$ if there is any stop point of $f$ within a certain distance $\psi$ from the
source and destination of $u$, i.e.,
$\var{dist}(u.p_1, f)\leq \psi \land
\var{dist}(u.p_{|u|}, f)\leq
\psi$.
Here, $\psi$ can be set based on the distance/range that a user can
cover on foot or by other means for availing the transportation
facility.
In such cases, we can define the Boolean service function $S(u,f)$ as:

\begin{math}
S(u,f) =
\begin{cases}
1 & {\text{if }} u {\text{ is served by }} f \\
0 & {\text{otherwise.}}
\end{cases}
\end{math}

\myparagraph{Scenario 2} In non-binary cases where $u$ can be served partially by $f$, the service can be computed based on the number of points in $u$ that can be served by $f$, $\var{scount}(u,f)$ (e.g., the number of POIs that can be visited by a tourist) as described in Scenario 2. Then the service value of $f$ is calculated as:
$S(u,f) = \frac{\displaystyle \var{scount}(u,f)}{\displaystyle |u|}$.

\myparagraph{Scenario 3} When the interest is in
maximizing the length of $u$ served by $f$, $\var{slength}(u,f)$ (e.g., the length of journey with advertisement display), the service value is calculated as: $S(u,f) =
\frac{\displaystyle \var{slength}(u,f)}{\displaystyle \var{length}(u)}$, where $\var{length}(u)$ is the total length of $u$.
Note that the length of two trajectories with the same number of
points can be different based on the length of the segments between
those points.

\subsection{Service value for the set of users}
As the objective of a facility is to maximize the service to $U$, the service value of a facility $f$ for $U$ is calculated as: 
\vspace{-4pt}
\begin{equation}
\var{SO}(U,f) = \sum\limits_{u \in U} S(u,f)
\end{equation}
For a collection $F^{\prime}$ of facilities, where
$F^{\prime} \subseteq F$, we can generalize the service value function as follows:

\vspace{-8pt}
\begin{equation}
\var{SO}(U,F^{\prime}) = \sum\limits_{u \in U} \var{AGG}_{f \in
F^{\prime}} S(u,f)
\end{equation}

Since a user can be served by more than one facility in
$F^{\prime}$, we only consider the service once if the same service
is provided by more than one facility.
The function $\var{AGG}$ takes this issue into account by 
aggregating the services provided by each $f \in F^{\prime}$ to $u$.

\myparagraph{Problem definition}
Based on the above definitions, we formally define our trajectory queries as follows.

\begin{definition} $($k$MaxRRST)$.
Given a set $U$ of user trajectories, a set $F$ of facilities, a positive integer $k$, and a service value function
$\var{SO}(\cdot)$, the $k$MaxRRST query returns the top-$k$ facilities $F^{\prime}$ from $F$ such that  $\forall f^\prime \in F^{\prime}$, $\forall f \in F\setminus F^{\prime},$ $\var{SO}(U, f^\prime) \geq \var{SO}(U ,
f)$.
\end{definition}

\begin{definition}
(Max$k$CovRST).
Given a set $U$ of user trajectories, a set $F$ of facilities, the group size $k$, and a service value function $\var{SO}(\cdot)$, let the set $\var{SG}_k$ be all possible subgroups
of size $k$ from $F$.
The Max$k$CovRST query returns a subgroup $\var{sg} \in \var{SG}_k$
of facilities such that for any other subgroup $\var{sg}^{\prime}
\in \var{SG}_k \setminus \{\var{sg}\}$, $\var{SO}(U, sg) \geq 
\var{SO}(U, \var{sg}^{\prime})$.
\end{definition}
\section{Trajectory Quad-tree (TQ-tree)}
\label{index}
The key observation behind our proposed indexing technique is, the
trajectories whose points (e.g., start points and end points) are
co-located, are likely to use the same facility.
Thus such trajectories should be stored together in the index.
Based on this observation we present a novel index, denoted as the
Trajectory Quad (TQ) tree, where trajectories with \emph{close
spatial proximity and similar orientation} are grouped and stored
together in an effective way.
For simplicity, we first describe the details of the index for
trajectories with two endpoints (source-destination), and later we
present the generalized structure for trajectories with any number of
points.
A two-level indexing is applied to index the trajectories in
a TQ-tree. We explain the index construction process and the rationale
behind each step in the following.

\myparagraph{Hierarchical organization} The space is recursively partitioned
to group spatially similar trajectories together.
Specifically, a quadtree structure is employed to partition the
space. Each node $E$ of the quadtree, denoted as a q-node is associated with 
a pointer to a list $\var{UL}(E)$ of user trajectories.

If $E$ is a leaf node, $\var{UL}(E)$ contains the intra-node trajectories, i.e., the trajectories whose both endpoints reside in
$E$.
Otherwise, $\var{UL}(E)$ consists of the
inter-node user trajectories, i.e., trajectories whose two endpoints
reside in two immediate child nodes of $E$. A node of the quadtree is partitioned until there is no such inter-node trajectories left to be stored with that node, or contains at most $\beta$ number of intra-node trajectories. Here, $\beta$ corresponds to the size of a memory block (or a disk
block for a disk-resident list $\var{UL}(E)$).

With each q-node $E$, an upper bound, $s_{ub}$ of the service value is stored for the trajectories stored in the subtree rooted at $E$. For Scenario 1, $s_{ub}$ of $E$ is the total number of user trajectories, for Scenario 2, $s_{ub}$ is the total number of points of the user trajectories, and for Scenario 3 $s_{ub}$ is the total length of the user trajectories stored in the sub-tree rooted at $E$, respectively.

\begin{figure}[ht]
	\hskip-0.18cm\begin{tabular}{ c c }
        \includegraphics[width=2in, height=1.55in]{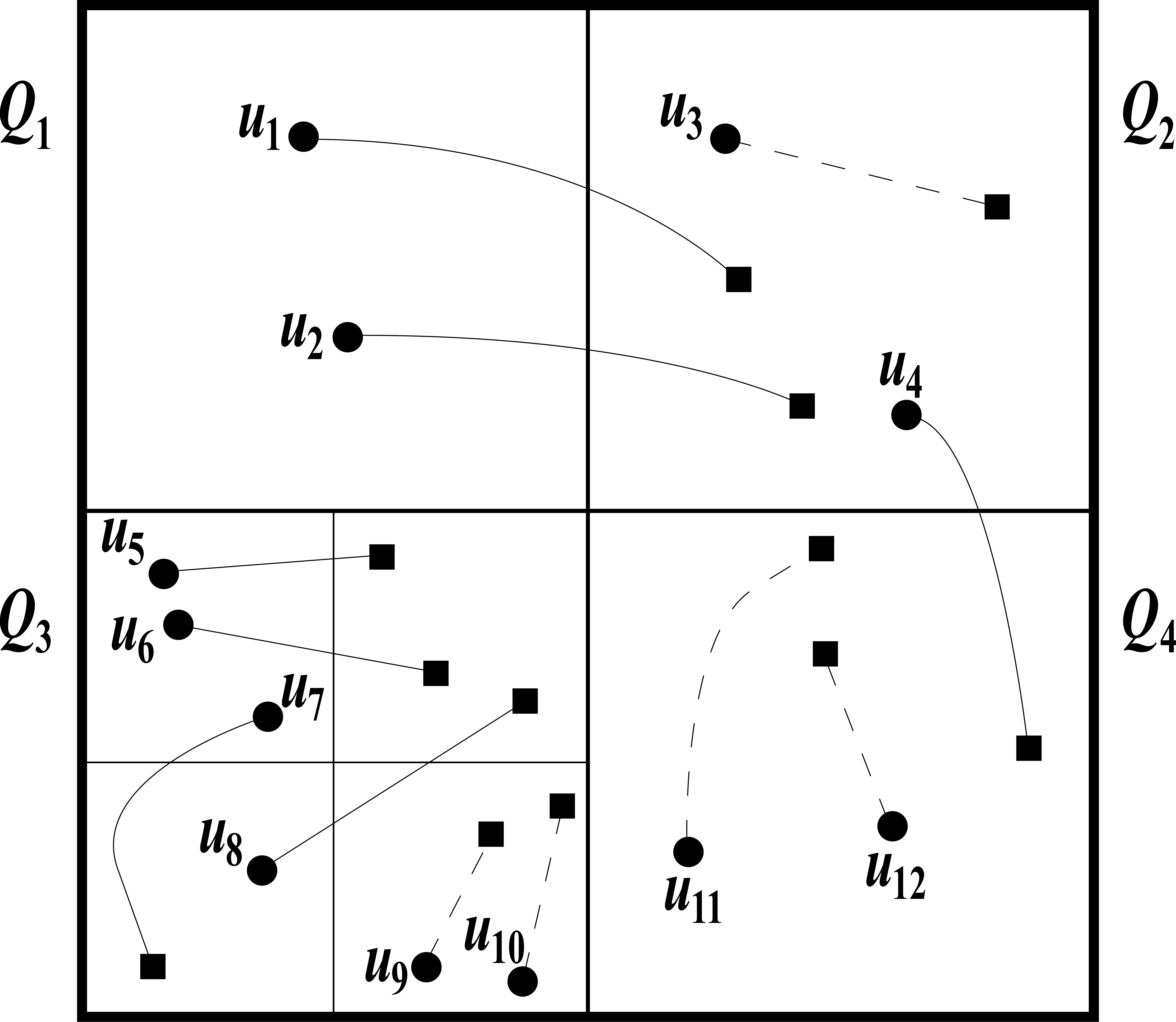} &
        \raisebox{0.22cm}{\includegraphics[height=1.5in]{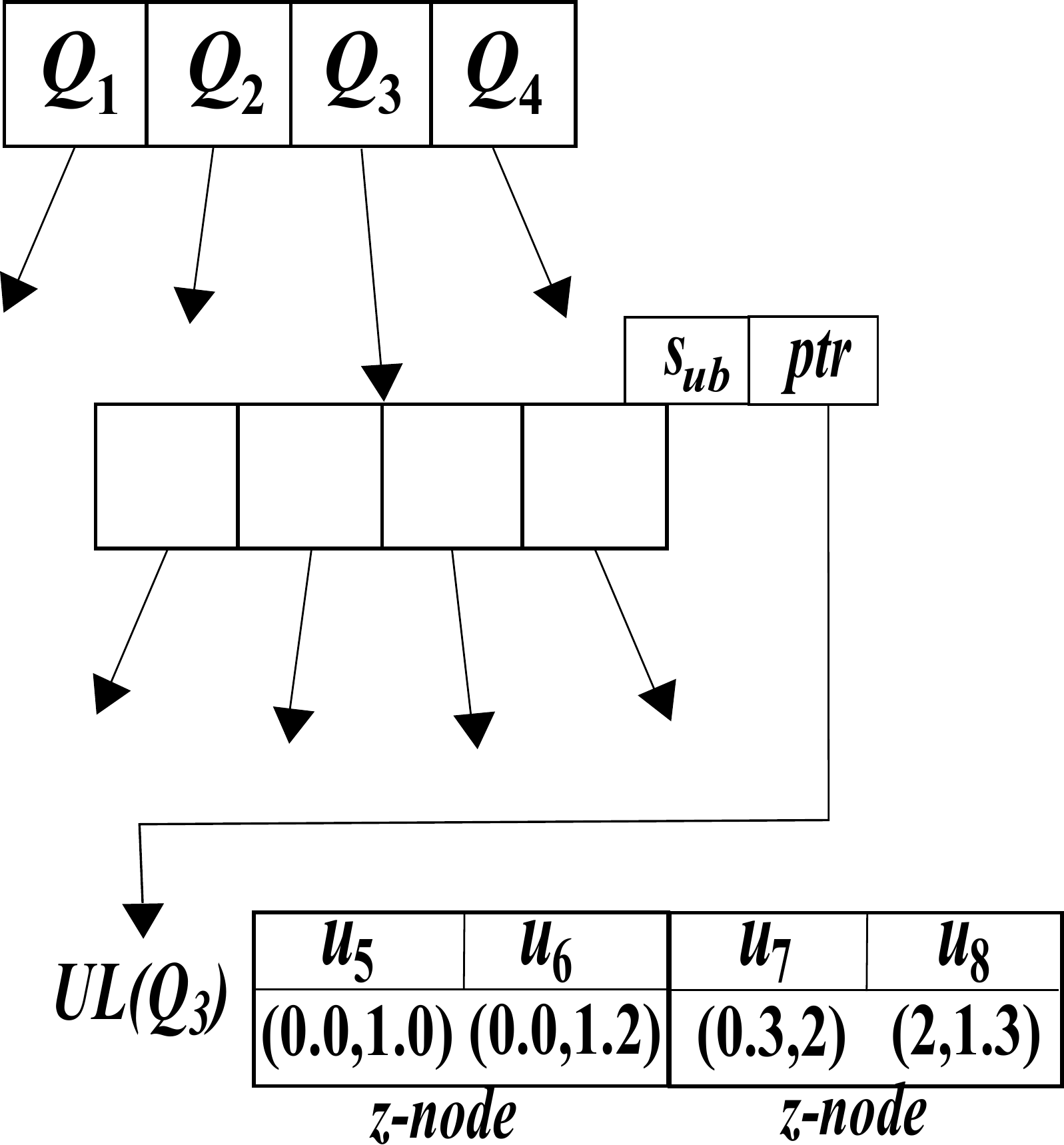}}\\
	\end{tabular}
	\vspace{-6pt}
    \caption{A TQ-tree structure for trajectories.}
    \label{fig:tqtree}
\end{figure}

\begin{figure}[ht!]
	\hskip-0.22cm \begin{tabular}{ c c c }
        \includegraphics[width=1.02in, height=0.9in]{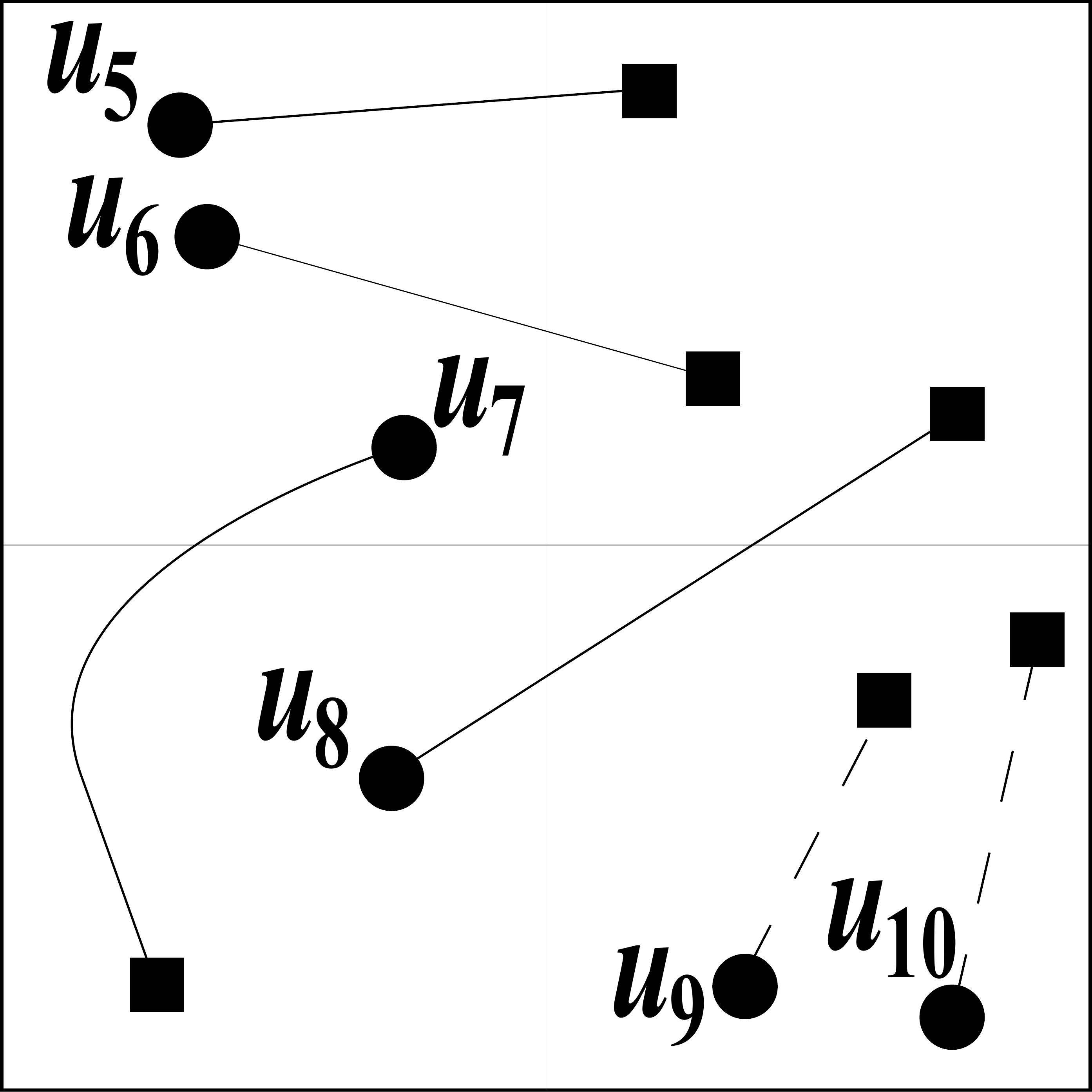} &
        \includegraphics[width=1.02in, height=0.9in]{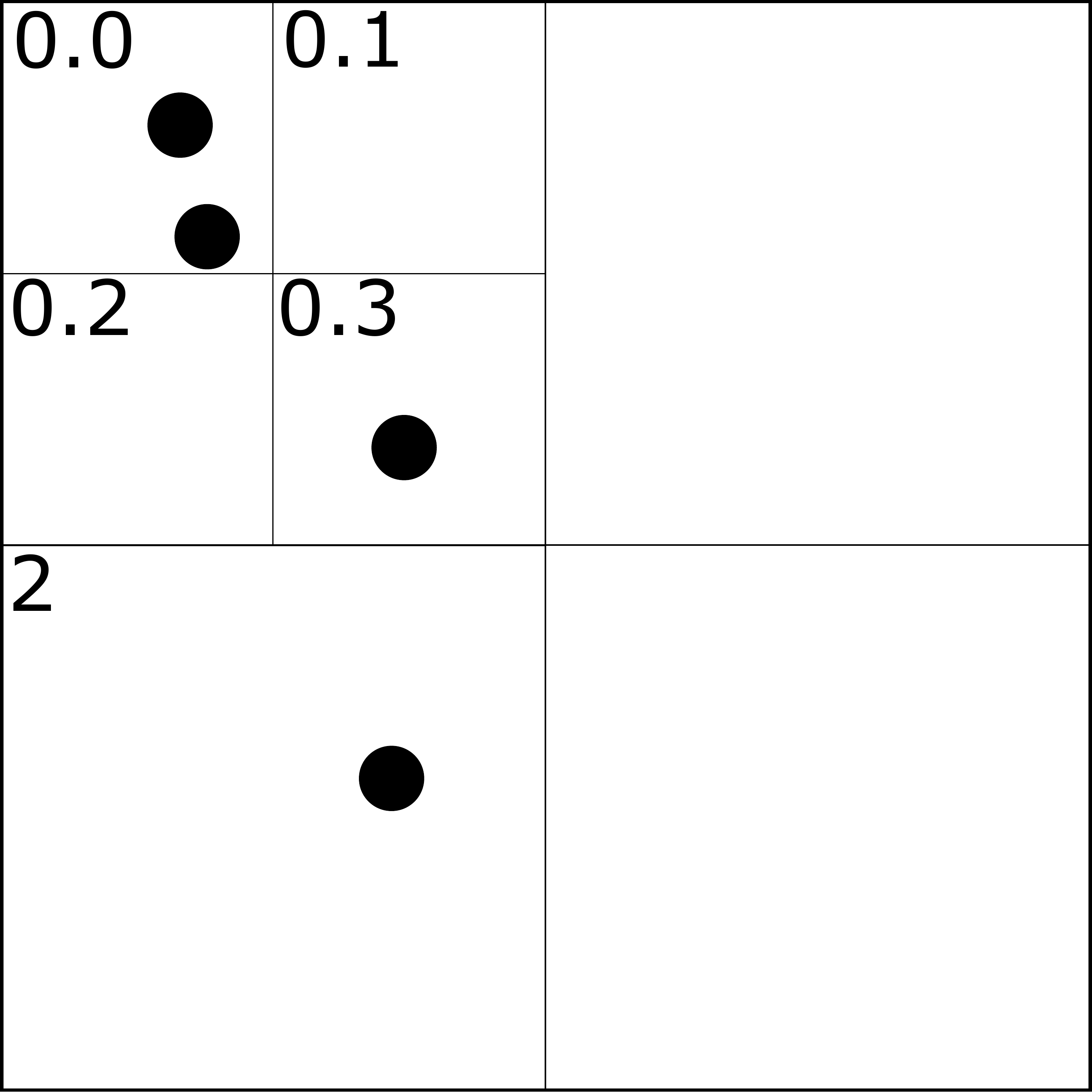} &
        \includegraphics[width=1.02in, height=0.9in]{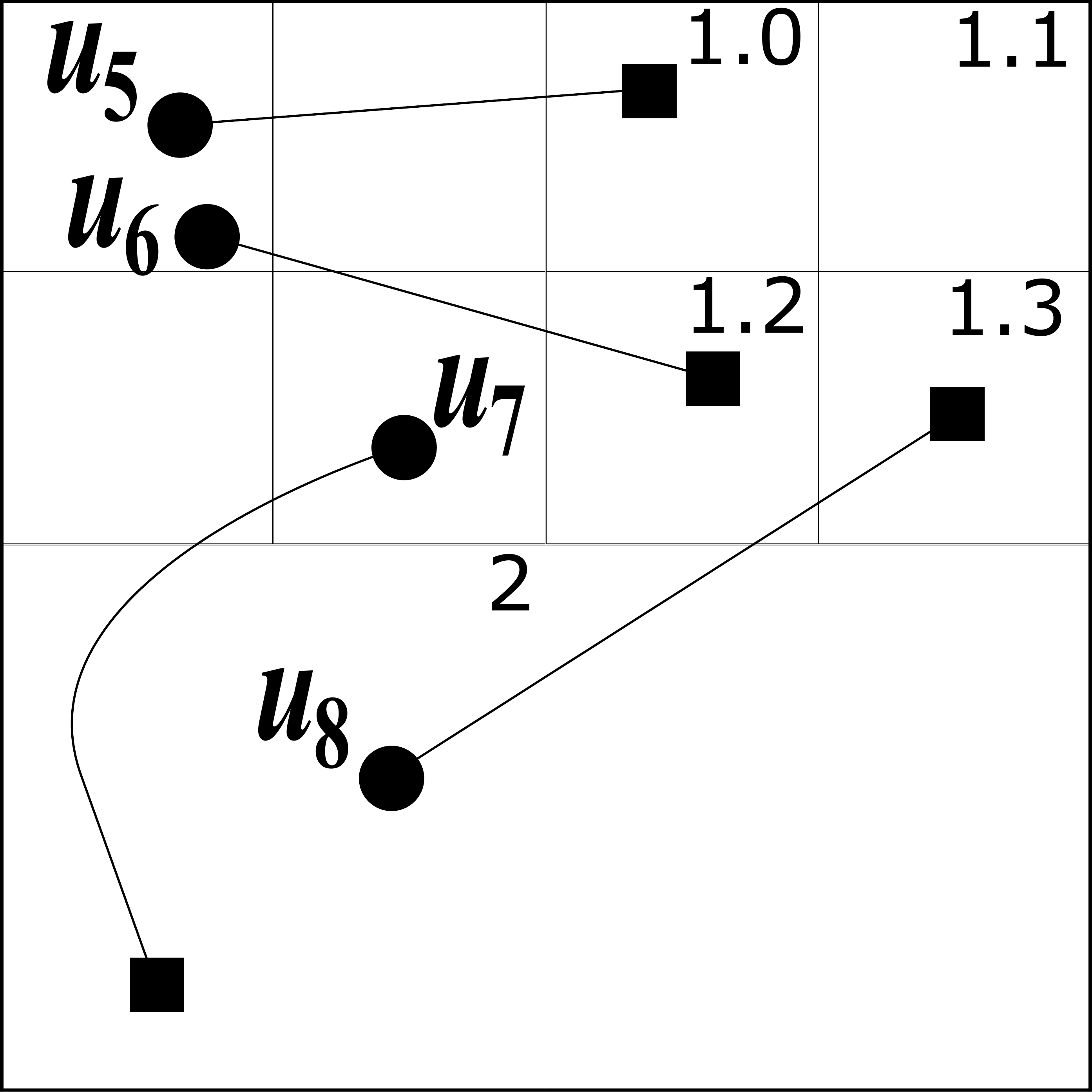} \\	(a) & (b) & (c) \\
	\end{tabular}
	\vspace{-6pt}
    \caption{(a) Inter-node trajectories of $Q_3$ (solid lines), (b) z-ordering of start points, (c) z-ordering of end points}
    \label{fig:znode}
    \vspace{-18pt}
\end{figure}

As mentioned in prior work~\cite{WangZZS15}, one of the major challenges of
indexing trajectories is in organizing the trajectories with
different lengths.
Unlike traditional spatial hierarchical indexing, where only the
leaf nodes contain the data, we store the trajectories in both leaf
and non-leaf nodes.
In this hierarchical organization, longer
trajectories are more likely to be stored in upper level nodes and
shorter trajectories in lower level nodes.
Such an organization will later facilitate efficient pruning and service (either
partial or complete) calculations for both longer and shorter
trajectories.

\begin{exmp}
{\textit {
Figure~\ref{fig:tqtree} shows an example TQ-tree for the user
trajectories, $\{u_{1}, \dots, u_{12}\}$, where let $\beta = 2$.
The space is first divided into four quadrants, $Q_1, \dots,Q_4$. As $Q_4$ only contains $\beta$ intra-node trajectories, and the trajectories in $Q_1$ and $Q_2$ are stored as the inter-node trajectories of the root node of the quadtree, these q-nodes are not partitioned further. The q-node $Q_3$ is further divided into four quadrants. The inter-node trajectories of $Q_3$ are $u_5,\dots,u_8$, and the partitioning terminates.}}
\end{exmp}

\begin{figure}[t]
	\hskip-0.4cm\begin{tabular}{ c c }
        \includegraphics[width=1.65in, height=1.35in]{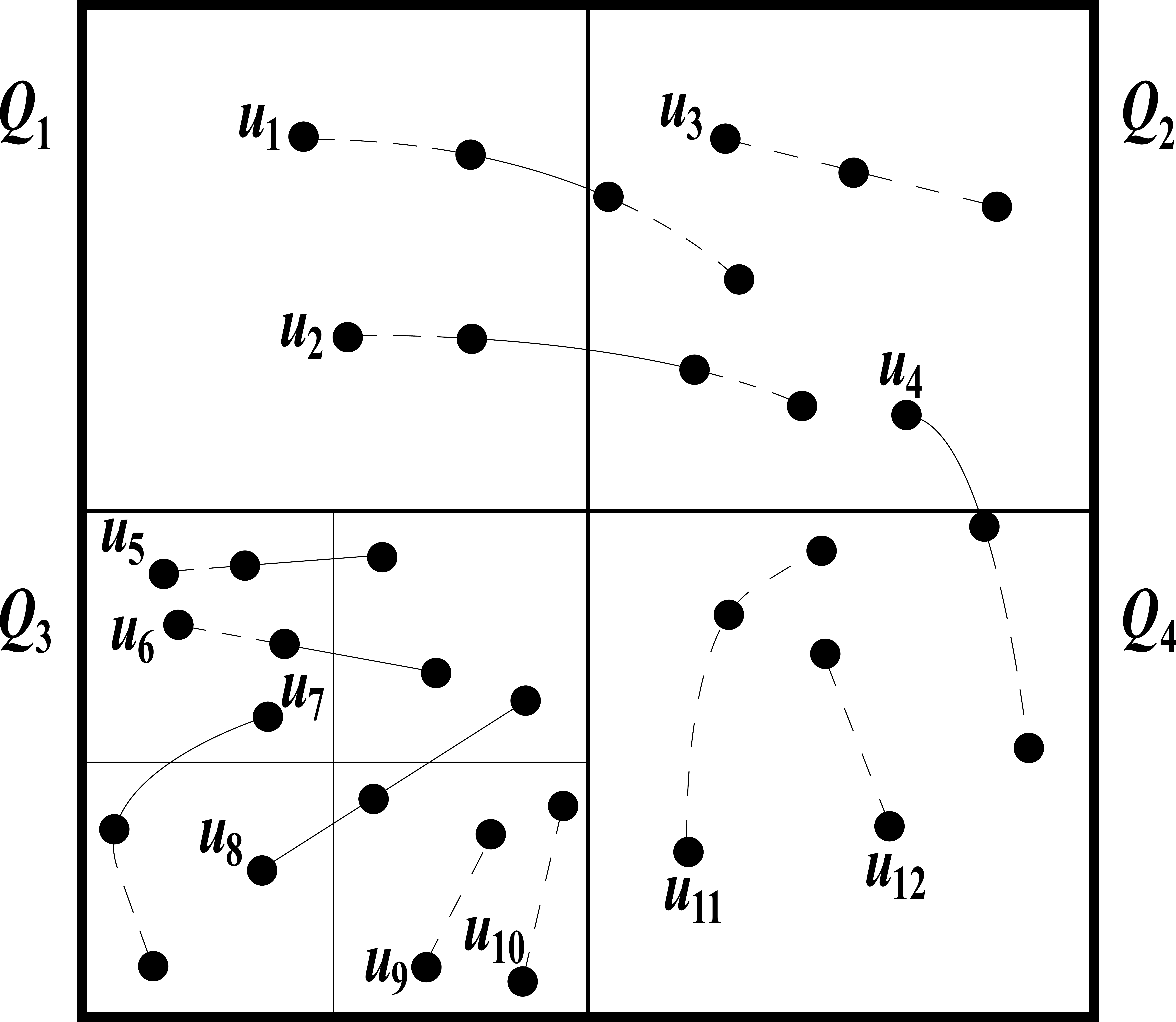} &
        \includegraphics[width=1.65in, height=1.35in]{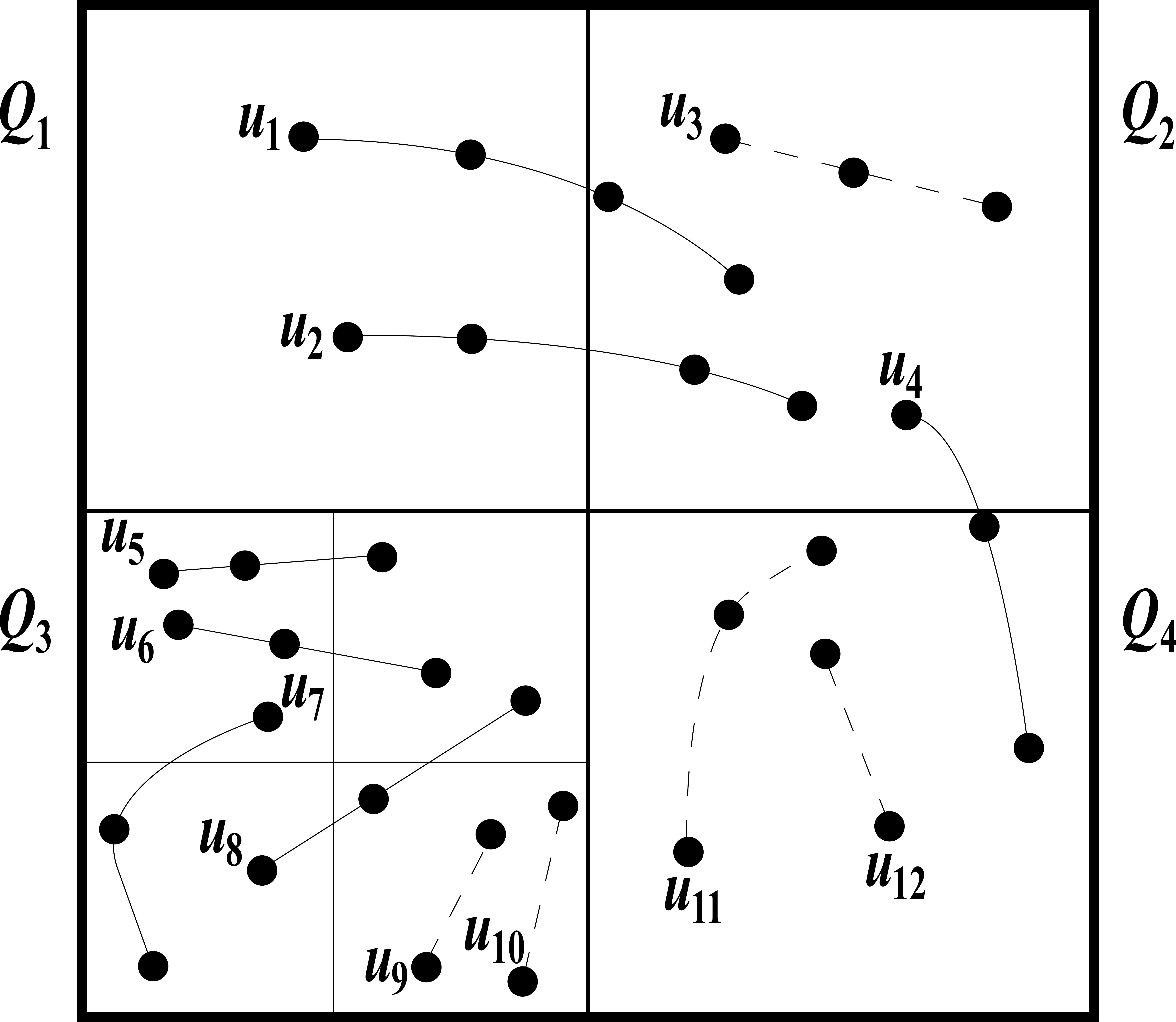}\\
	(a) Segmented TQ-tree & (b) Full-trajectory TQ-tree\\
	\end{tabular} 
	\vspace{-6pt}
    \caption{Multiple-point Trajectories in a TQ-tree (solid lines are inter-node, and dashed lines are intra-node segments).}
    \label{fig:gentraj}
    \vspace{-18pt}
\end{figure}

\myparagraph{Ordered bucketing using z-curve:} Depending on the application
scenarios and the user travel patterns, the list of trajectories in a
q-node can be quite large.
For example, if there are many users who travel everyday from the
same suburb to the city, these user trajectories may all fall under a
particular q-node.
Thus, a straightforward approach to store these trajectories as a
flat list may result in poor query processing performance.
Therefore we use a space filling curve, specifically a Z-curve
(Morton order) to order the trajectories such that the trajectories
with \emph{close spatial proximity and similar orientation} are
grouped together into a single ``bucket".
The list $\var{UL}(E)$ of each q-node is arranged as a sorted list of
buckets, where the trajectories in each bucket is also sorted by
their z-ordering. Here each bucket is referred to as a \emph{z-node}.

Specifically, for each q-node $E$, (i) we first apply the z-ordering on
the start points of the user trajectories in $\var{UL}(E)$.
The space enclosed by $E$ is partitioned until each partition contains at most $\beta$ start points of user trajectories. 
(ii) Then, we partition the space based on the end points of the user
trajectories, where each partition can contain a maximum of $\beta$
end points. Also, if multiple trajectories have the same z-id for their start points, the space is partitioned until the end point of each such trajectory is assigned a different z-id. This step enables us to distinguish between the trajectories with the co-located start points. 
(iii)
Based on the z-order numbers assigned to each points of user
trajectories, we keep them in a sorted bucket list, where each bucket
can contain at most $\beta$ trajectories.

If each trajectory is an ordered sequence of points, then we order the trajectories based on the
starting point first, and if two trajectories have the same z-orders,
we order them based on their second points, and so on. If the trajectories are defined as non-ordered
sequence of points, then we order the points of a trajectory based on the z-order, and then apply the aforementioned procedure to
sort the trajectories.

\begin{exmp}
{\textit{
Figure~\ref{fig:znode} shows the construction process of z-nodes.
The q-node $Q_3$ points to $\var{UL}(Q_3)$ of four inter-node trajectories, $u_5, u_6, u_7, u_8$. To obtain the z-ordering, the space of $Q_3$ is partitioned based on the start points of the trajectories, and each partition is assigned a z-id where a partition can
have at most $\beta = 2$ start points (Figure~\ref{fig:znode}(b)). As an example, the start points of both $u_5$ and $u_6$ have $0.0$ as
their z-ids.
Next, we apply the same partitioning strategy on the end points. The end points of $u_5$, $u_6$, $u_7$, and $u_8$ are
assigned z-ids $1.0$, $1.2$, $2.0$, and $1.3$, respectively and the partitioning terminates. Finally, a pair of z-ids for each trajectory is kept in z-nodes, each of size $\beta$ (Figure~\ref{fig:tqtree} (right)).}}

\end{exmp}

\subsection{Generalization of the Index} \label{sec:index_gen}

So far we explained our index for trajectories with
two points (source and destination), which serves only a subset of
queries described in Scenario 1.
To serve other types applications that require maintaining a sequence
of points in each trajectory where a trajectory can be served partially, we generalize our index as follows.
We propose two approaches: a segmented approach, and a full-trajectory
approach.

\myparagraph{Segmented approach:} We segment each trajectory into a sequence
of pairs of points, and then for each pair of points (segment) we apply the same
strategies described above. 

Here, indexing each segment of the trajectories in hierarchy and ordered lists will enable us to calculate the total and the partial score of service (explained later in Section~\ref{algorithm}. This process is depicted in Figure~\ref{fig:gentraj}(a).

\myparagraph{Full trajectory approach} In some applications, we need to
consider the entire trajectory contiguously as the objective function may
need to quantify the coverage of an individual user trajectory that
is served by the facilities.
For these applications, we propose a full-trajectory
approach, where we store a trajectory in the q-node
at the lowest level of the quadtree that fully contains the entire
trajectory.
In an intermediate q-node, all inter-node trajectories are sorted
using z-orders, and in a leaf q-node intra-node trajectories are
stored using z-orders (as described previously).
This scenario is depicted in Figure~\ref{fig:gentraj}(b).

Note that we use a quadtree to partition the space and then organize
the trajectory information in each quadtree cell using space filling
curves.
The main reason for using quadtree is that it supports efficient
frequent updates. Moreover, since a quadtree partitions the space into disjoint cells, we can apply z-orders to generate unique ids for the points in a trajectory.

\subsection{Index Storage Cost}
The space requirement of the hierarchical component of the TQ-tree
includes storing the nodes of the quadtree, specifically,
$\mathcal{O}(h)$, where $h$ is the height of the tree.
If only the source and destination points of each trajectory is of
interest, then a trajectory is stored exactly once in an appropriate
node of the TQ-tree.
Thus the total size of the user trajectory lists in all nodes,
$\sum_{E \in TQ-tree}UL(E)$, is at most the total number of user
trajectories $|U|$.
The same storage costs apply for the \emph{full trajectory approach}
as well.

In the generalized TQ-tree, each segment of a user trajectory is
stored in an appropriate node of the TQ-tree.
The total number of segments of a trajectory $u$ is $|u| - 1$, and a
segment is stored exactly once.
Thus the total size of the user trajectory lists in all the nodes,
$\sum_{E \in TQ-tree}UL(E)$ in the generalized TQ-tree is $\sum_{u
\in U} |u| - 1$.

\subsection{Updating the Index}
Since the TQ-tree uses a regular space partitioning scheme, to
insert a new user trajectory, $u$, we can quickly identify the
corresponding q-node to which $u$ belongs to in $\mathcal{O}(h)$ time.
Then, $u$ needs to be inserted in an appropriate z-node of the user
trajectory list.
If the number of points in the corresponding z-node does not exceed
the threshold $\beta$, no further partitioning is needed.
The points of $u$ are assigned the appropriate z-ids, and inserted in
the sorted user trajectory list.
Otherwise, the corresponding z-node is partitioned and the z-ids are 
assigned to the points of $u$.
Since the z-ids of the existing user trajectories in that z-node may
change, we may need to re-assign z-ids to the trajectories.
This re-assignment needs to be done for at most $\beta$
trajectories in that z-node.

\section{Processing $k$MaxRRST Queries}
\label{algorithm}

\setlength{\algomargin}{1.2em}
\begin{algorithm}[t]
    \caption{evaluateService($Q$,$f$)}
    \label{alg:service_f}
    \begin{smaller}
    \KwIn{A q-node $Q$ of {\tqtree}, a facility component $f$}
    \KwOut{Service value $so$ of $f$ for users in subtree rooted at $Q$}
    $so \gets 0$ \\
    \lIf {$f = \varnothing$} {\Return 0 \DontPrintSemicolon}  \label{a1_end1}
    \uIf {$Q$ is a $leaf$} {\Return evaluateNodeTrajectories($Q, f$)} \label{a1_end2}
    $Q_{children} \gets children(Q)$ \\
    $f_{children} \gets$ intersectingComponents($Q_{children}, f$) \label{line:a1_fi} \\
    \For{$q_c \in Q_{children}$, $f_c \in f_{children}$\label{line:a1_for} } 
    {
    	$so \gets so +$ evaluateService($q_c,f_c$) \label{line:a1_recur}
    }
    $so \gets so +$ evaluateNodeTrajectories($Q, f$)\\
    \Return $so$
    \end{smaller}
\end{algorithm}

\begin{figure}
\centering
\subfloat[]{\includegraphics[height=0.15\textwidth]{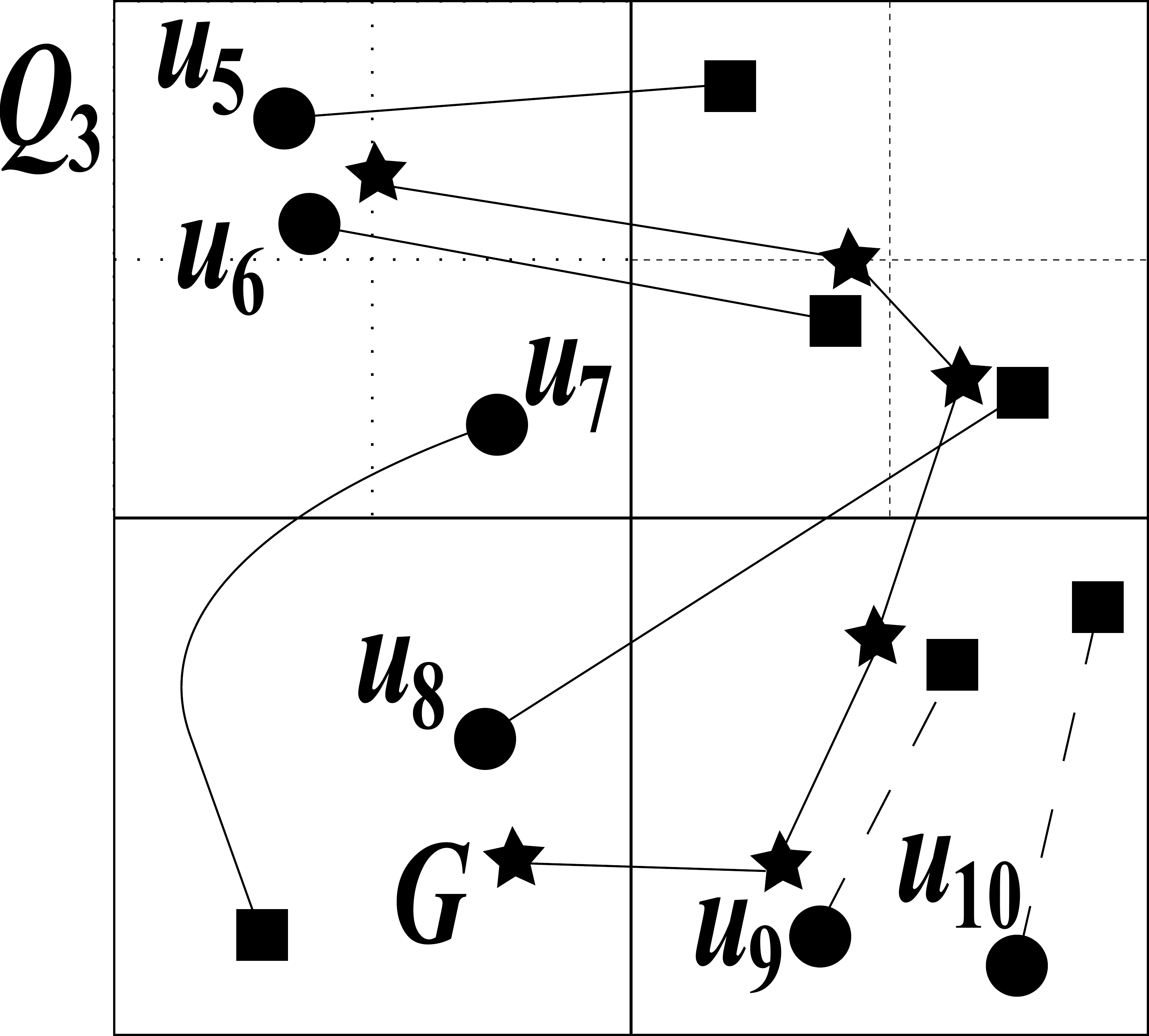}\label{fig:ow}} \hfill
\subfloat[]{\includegraphics[ height=0.15\textwidth]{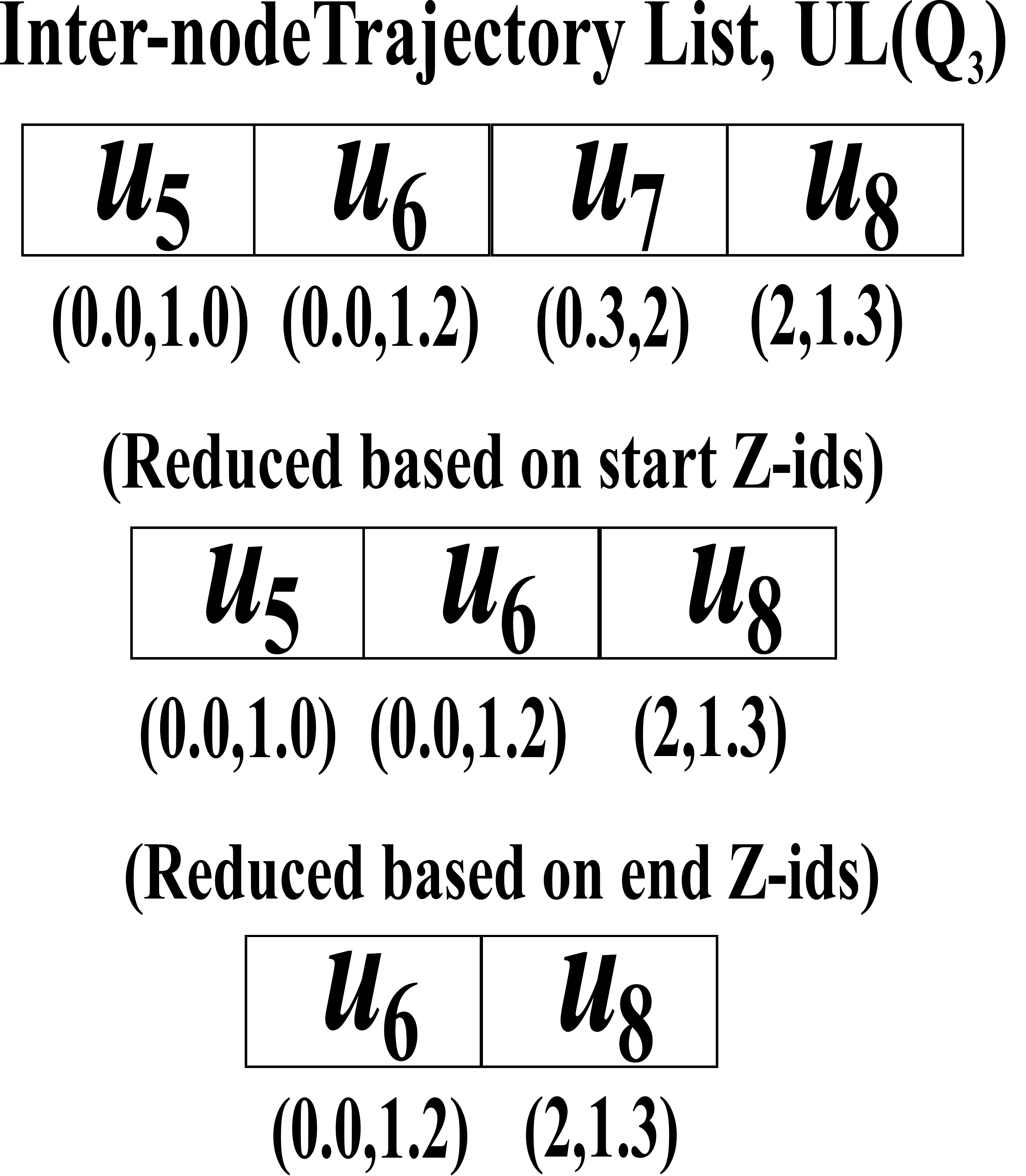}\label{fig:tw}} \hfill
\subfloat[]{\includegraphics[height=0.15\textwidth]{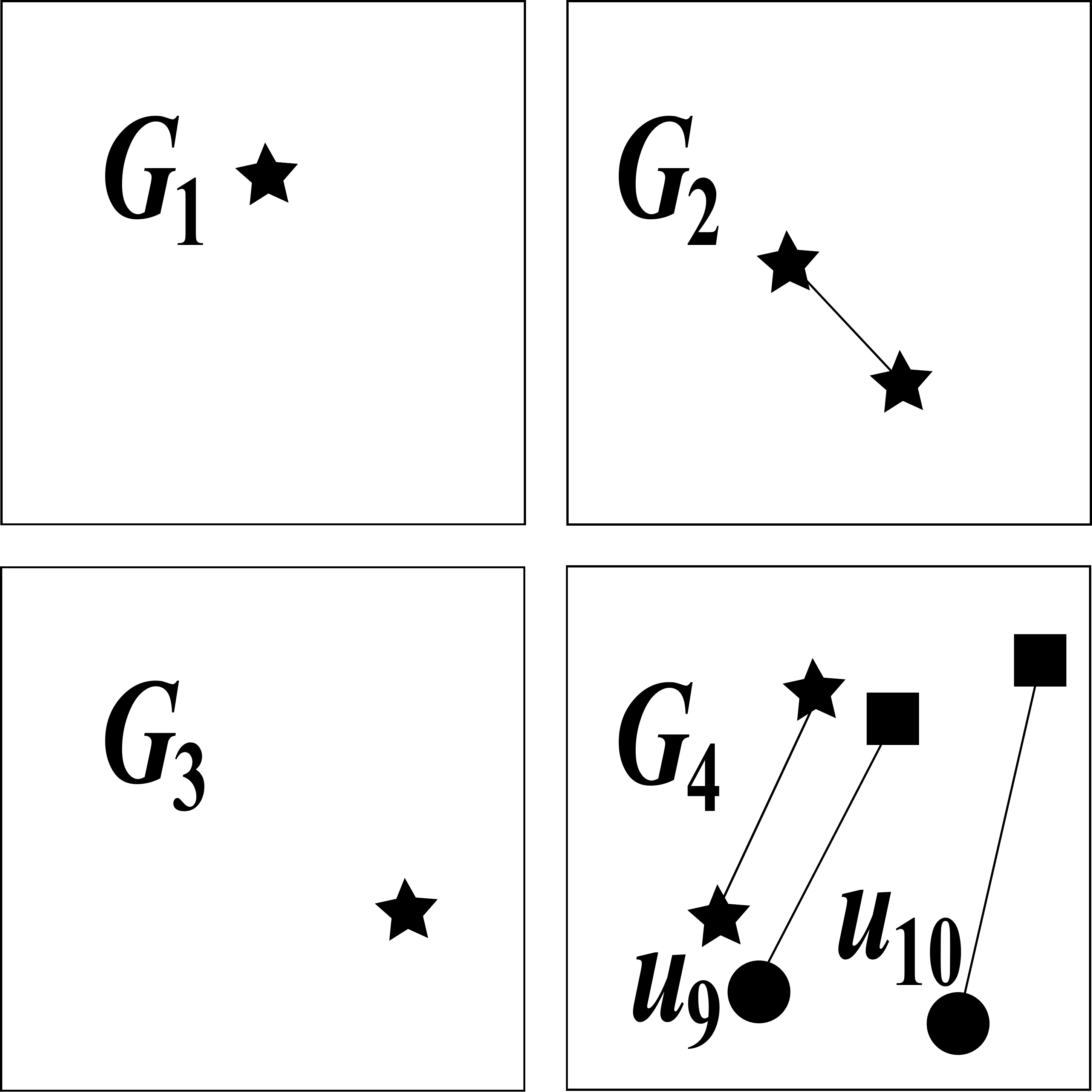}\label{fig:ow}} 
\vspace{-6pt}
\caption{(a) A q-node $Q_3$ with trajectories and facilities (b) Z-reduce for reducing $UL(Q_3)$ for $G$ (c) Recursive calls for subspaces with corresponding facility subgraphs, $G_1, G_2, G_3, G_4$}
\vspace{-22pt}
\label{fig:zreduce}
\end{figure}

In {\km} query, a user trajectory can be partially served by a
facility.
Thus, an efficient technique is needed to calculate the appropriate
service value of a facility for the set of user trajectories.
In this section, we first propose an efficient divide-and-conquer
algorithm to recursively divide a facility trajectory and traverse
only the necessary nodes of the {\tqtree} to calculate the service
value of the components of the facility in that subspace.
We apply a two-phase pruning technique using the {\tqtree}, where the
q-nodes are pruned first, and then the z-ordering of the
trajectories are used to further prune the z-nodes.
A merge step is evoked to check if the same user trajectory can be
served by the connected components of the same facility, and an upper
bound of the service value of that facility is updated from the
current state of exploration.
A best-first strategy is employed to explore the facilities
based on their estimated upper bounds of service values.
In this section, we first present our algorithm for computing the service value of a single facility $f \in F$.
Then we present our approach to find the top-$k$ facilities from $F$ with the maximum service value.

\subsection{Calculating Service Value: Divide-and-Conquer}

Algorithm~\ref{alg:service_f} shows the pseudocode for calculating
the service value of a facility $f \in F$.
Note that in our application scenarios, a user $u$ can be served by $f$ (partially or completely) if a point of $u$ is within
a threshold distance $\psi$ from any point of $f$. Thus, we cover $f$ with an extended minimum bounding rectangle (EMBR)
that includes the serving area of $f$.
However, without loss of generality, we use the term $\var{EMBR}$ and $f$
interchangeably when we match users with $f$.

Initially, the function $evaluateService(\cdot)$ is called with the root node $Q$ of the {\tqtree} for $f$. First, it finds the relevant child q-nodes of $Q$ that intersect
with $f$ (or EMBR of $f$) in the function
$\var{intersectingComponents}(\cdot)$ (Line~\ref{line:a1_fi}).
If a child q-node does not intersect, that
q-node can be safely pruned.
Otherwise, the $\var{EMBR}$ of $f$ is divided into four equal subspaces.
For each unpruned child q-node $q_c$ of $Q$ and the corresponding intersecting components of $f$, the function $\var{evaluateService}$ in
Algorithm~\ref{alg:service_f} is recursively called (Line~\ref{line:a1_recur}).

The recursive call terminates at two conditions: (i) If $f$ is empty, i.e., after division there is no
point left in that subspace that can server any user (Line~\ref{a1_end1}); and (ii)
When $Q$ is a leaf node. For a leaf node, the function $\var{evaluateNodeTrajectories}(\cdot)$ is called to compute the service value for the
intra-node trajectories in $UL(Q)$ of that node.

The function $\var{evaluateNodeTrajectories}$ is used to determine the service value that is increased
for serving the trajectories in $UL(Q)$. Algorithm~\ref{algo:nodeservice} describes that process.
A merge step is employed in this algorithm as the function
$\var{MakeUnion}$($f$) to check whether the same user trajectory can be
served by the same connected components of $f$.
Here, the connected components of $f$ are assigned unique
identifiers.
Next, we need to access the trajectories in $\var{UL}(Q)$ (that are stored as a sorted list of z-nodes according to z-order).
We apply a pruning technique using the z-order ids of the trajectory
points to get a list $T_r$ of a reduced size from $\var{UL}(Q)$ using the
$\var{zReduce}(\cdot)$ function in Lines~\ref{line:alg2_ul} - \ref{line:alg2_zr} (explained later).
For each user trajectory $t_i$ in the reduced list $T_r$, we
compute the service value gained for serving $t_i$ by $f$.

\setlength{\algomargin}{1.2em}
\begin{algorithm}[t]
    \caption{evaluateNodeTrajectories($Q$,$f$)}
    \label{algo:nodeservice}
    \begin{smaller}
    \KwIn{A q-node $Q$ of {\tqtree}, a facility component $f$}
    \KwOut{Service value $so$ of $f$ for trajectories stored in $Q$}
    $us \gets$ MakeUnion($f$) \\
    $T_q \gets$ $\var{UL}(Q)$ \label{line:alg2_ul}\\
    $T_r \gets$ zReduce($T_q, f$) \label{line:alg2_zr} \\
    $so \gets 0$ \\
    \For{$t_i \in T_r$}{
	$so = so + \text{serviceValue}(t_i,f)$
    }
    \Return $so$
    \end{smaller}
\end{algorithm}

Note that the evaluation of function $\var{serviceValue}(\cdot)$
varies across different applications.
For example, in Scenario 1, where we are only interested in serving
start and end points, $\var{serviceValue}(t_i,f)$ returns $1$ if both points of $t_i$ are within $\psi$ distance from any of the stop points of $f$ (or a connected component of $f$). Otherwise, $\var{serviceValue}(t_i,f)$ returns $0$.

The function \emph{zReduce} in Algorithm~\ref{algo:nodeservice} is
used to prune the inter-child trajectories that cannot contribute to the service value. The idea is to avoid searching the full list of inter-child
trajectories and reduce the list to a small relevant set of
trajectories based on the spatial properties of $f$ and z-ordering.
This function takes the inter node trajectory list and a component of
the facility as input.
It prunes the user trajectories based on the z-ids that the facility intersects.

\begin{exmp}
{\textit{
We explain the $\avar{zReduce}(\cdot)$ function with an example in
Figure~\ref{fig:zreduce}.
The figure shows a facility trajectory $G$, and a list $T_q$ of inter-node trajectories $\{u_5, u_6, u_7, u_8\}$ with start and end z-ids
\{(0.0,1.0), (0.0,1.2), (0.3,2), (2,1.3)\} of q-node $Q_3$.
Let $G$ can intersect nodes with z-ids $0.0, 0.1, 1.2, 1.3, 2, 3$ fully or partially, i.e., the stop points in $G$ are within $\psi$ distance to serve fully or some
portions of these z-nodes.
Thus, trajectory $u_7$ is pruned since its start z-id $0.3$ is not covered by $G$.
In this step, we get a reduced list $\{u_5, u_6, u_8\}$ with z-ids
\{(0.0,1.0), (0.0,1.2), (2,1.3)\}.
Next we look at the z-ids of end points for pruning the list further. Here, $u_5$ is pruned since its end z-id is $1.0$, and we get the final
reduced list $\{u_6, u_8\}$.
Thus in two steps the inter-node trajectory list $\{u_5, u_6, u_7,
u_8\}$ gets reduced to $\{u_6, u_8\}$.
After reducing $\var{UL}(Q)$ to $\{u_6, u_8\}$ (Figure~\ref{fig:zreduce}(b)), we divide $G$ into four
sub-spaces, $G_1, G_2, G_3, G_4$, and evaluate the service values of
these subspaces by calling Function $\avar{evaluateService}(\cdot)$
(Figure~\ref{fig:zreduce}(c)).}}
\end{exmp}

\setlength{\algomargin}{1.2em}
\begin{algorithm}[t]
    \caption{TopKFacilities($F, k$)}
    \label{algo:bestk}
    \begin{smaller}
    \KwIn{A set of facilities $F$, a positive integer $k$} 
    \KwOut{Top $k$ facility collections $F^{\prime}$} 
    Initialize a max-priority queue $PQ$; $F^{\prime} \gets \varnothing$ \label{a3_init2}\\
    \For{$f_i \in F$}
    {
	$Q \gets$ containingQNode($f_i$) \label{a3_s1}\\
	$\var{qfPair} \gets$ makePair($Q$, $f_i$)\\
	Initialize a state $S$ with id i \\
	Insert($S.{\var{qflist}},\var{qfPair}$)\\
	$S.\var{aserve} \gets 0$; 
	$S.\var{hserve} \gets Q.s_{ub}$\\
	$\var{fserve}(S) \gets S.\var{aserve} + S.\var{hserve}$  \label{a3_s2}\\
	$\var{PQ}$.push($S$, $\var{fserve}(S)$)\\
    }
    \Repeat{$|F^{\prime}|  = k$} {
	$S \gets \var{PQ}$.pop() \label{a3_repeat1}\\
	\eIf {$S.\var{qflist} = \varnothing$} 
	{
		Insert$(F^{\prime},S.\var{id})$
	}
	{$S_{\ssvar{new}} \gets$ relaxState($S$)\\
	$PQ.$push($S_{\ssvar{new}}$, $\var{fserve}(S_{\ssvar{new}})$)\\}\label{a3_repeat2}}
     
    \Return $F^{\prime}$
    \end{smaller}
\end{algorithm}

\myparagraph{Algorithm for multiple-point trajectories} In Algorithm 2, the
$\var{serviceValue}(\cdot)$ function returns a normalized score that
is achieved for serving multiple-point trajectory $t_i$ by $f$.
The normalized score depends on the requirements of the applications (e.g., Scenario 2 or 3): one may want to count
the number of points in $u$ served by $f$, or find a summation of the segments of $u$ served by $f$. To accommodate such applications, the service value calculation changes accordingly.

Based on our above algorithms, we now propose an approach that
finds the top-$k$ facilities from a set $F$ of facilities.

 \subsection{Finding Top-$k$ Facilities}

Algorithm~\ref{algo:bestk} shows the pseudocode for finding the top-$k$ facilities from $F$.
The key idea is to apply a best-first technique to
explore facilities based on their predicted upper bounds of the
service values.
The upper bound of the service value, $\var{fserve}$ of a
facility (or a collection of facilities) is computed by combining the
value of the \emph{actual} service function, $\var{aserve}$, from the
current state of exploration and the optimistic value of the service
function, $\var{hserve}$, which is estimated based on a \emph{heuristic},
i.e., the maximum service value that can be achieved by further exploration of the facility.

For each facility $f_i \in F$, we
maintain a tuple $S$ to preserve its current state of exploration. $S$ contains the following information: identity $id$ of the facility, a list $\var{qflist}$ of $\langle$q-node, facility-component$\rangle$
pair that overlaps with each other, the actual service value $\var{aserve}$ of the facility based on the actual number of
users served so far by the current state of exploration, and the maximum
value of the service $\var{hserve}$ that can
be achieved by $f_i$ in the remaining parts of the exploration.
We maintain a max-priority queue, $\var{PQ}$ of such tuples $S$ according to the
upper bound values, $\var{fserve}$, where, $\var{fserve} =
\var{aserve} + \var{hserve}$ (i.e., the summation of the actual service
value achieved and the upper bound of the service value that can be
achieved by the facility).

\setlength{\algomargin}{1.2em}
\begin{algorithm}[t]
    \caption{relaxState($S$)}
    \label{algo:relaxing}
    \begin{smaller}
    \KwIn{Current state of a collection of facilities, $S$}
    \KwOut{Relaxed state $S_{r}$}
    Initialize $S_r$\\
    \For{each pair $(Q,f) \in S.{\avar{qflist}}$}{
	$S_{r}.\var{aserve} \gets S_{r}.\var{aserve}$ + evaluateNodeTrajectories($Q$, $f$) \label{a4_as}\\
	{$Q_{\ssvar{children}}$} $\gets$ {children($Q$)}\\
    	{$f_{\ssvar{children}}$} $\gets$ {intersectingComponents($Q_{\ssvar{children}}$, $f$)}\\
    	\For {$q_c \in Q_{\assvar{children}}, f_c \in f_{\assvar{children}}$\label{a4_loop1}}{
	    \If {$f_c \ne \varnothing$}{
	    	$\var{cqgPair} \gets$ makePair($q_c$, $f_c$)\\
	    	Insert($S_{r}.\var{qflist},\var{cqgPair}$)\\
	    	$S_{r}.\var{hserve} \gets S_{r}.\var{hserve}$ + $q_c.s_{ub}$ \label{a4_loop2}\\
	    }
    	}
    }
    \Return $S_{r}$
    \end{smaller}
\end{algorithm}

The result set $F^{\prime}$ and the priority
queue $\var{PQ}$ is initialized (Line~\ref{a3_init2}) as empty.
The states are initialized for each facility, and inserted in $\var{PQ}$ (Lines~\ref{a3_s1} - \ref{a3_s2}).
Function $\var{containingQNode}(f_i)$ returns the smallest q-node,
$Q$ that contains $f_i$ (Line~\ref{a3_s1}).
A pair is formed with the facility component $f_i$ and the corresponding $Q$.
The pair $(Q,f_i)$ is inserted in $\var{qflist}$ of $S$.
We initialize $\var{aserve}$ with 0 (as no user trajectories has been
matched with $f_i$) and $\var{hserve}$ with the upper bound $s_{ub}$ of the service values stored with the node $Q$ in the {\tqtree}. As described in Section~\ref{index}, depending on the application, the upper bound of serve is different. For example, for scenario 1, $s_{ub}$ of a node $Q$ is the number of
trajectories contained in $Q$.
We insert the current state $S$ of $f_i$ along with the total upper bound of
the service value, $\var{fserve}(S)$, achieved so far by the
current state of facility exploration.

Next, we progressively explore user trajectories by
relaxing different parts of facility trajectories to find the top-$k$ facilities that maximize the service.
In each iteration (Lines~\ref{a3_repeat1} - \ref{a3_repeat2}), the facility component with the
maximum $\var{fserve}$ value is dequeued from $\var{PQ}$, and the
state is updated by relaxing the component through a function call
$\var{relaxState}(S)$ that explores the children of the corresponding
q-node, updates $\var{aserve}$ and $\var{hserve}$, and inserts the
new state into $\var{PQ}$.
If $\var{qflist}$ of the dequeued facility is empty, it implies that all components of this facility trajectory are
explored, thus the facility is added to the result set.
The process terminates when top-$k$ facilities are found, and the result list $F^{\prime}$ is returned.

\myparagraph{Relaxing state}
Algorithm~\ref{algo:relaxing} shows the pseudo-code of how to relax the state of a facility component. The input is the current state, the relaxed (or more
expanded) state is returned as output.

First, we initialize the variables of the new relaxed state $S_r$ as $S_r.id \gets S.id$, $S_r.aserve \gets  S.aserve$, $S_r.hserve \gets 0$, and $S_r.qflist \gets \varnothing$.
Next, for pair $(Q,f)$ of q-node and facility component in the $\var{qflist}$ of input state $S$, we expand the component with respect to the children q-nodes of $Q$.
In this expansion and update process, we update the value $\var{aserve}$,
i.e., number of users already served by adding the number of
inter-node trajectories of the corresponding q-node that can be served by $f$.

For this purpose, we compute the service value of that q-node by calling the evaluateNodeTrajectories($\cdot$) function, and add this value to $S_{r}.\var{aserve}$, as $S_{r}.\var{aserve}$ denotes the value of trajectories already served (Line~\ref{a4_as}).

Next we get the child q-nodes and corresponding components of the
facility. In the loop presented in Lines~\ref{a4_loop1} - \ref{a4_loop2} we update the list of $\langle$q-node,
facility component$\rangle$ pair for each of the child nodes and the maximum
value of service $S_{r}.\var{hserve}$
with the upper bound service value $s_{ub}$ stored in the child q-node (Line~\ref{a4_loop2}).
The outer loop terminates when we complete the same computation for
all the members of the $(Q,f)$ pair list of the current state.
Finally the relaxed state $S_r$ is returned.

\section{Processing Max$k$CovRST}
\label{maxkcov}
The Max$k$CovRST query is a variant of the maximum coverage problem, which is NP Hard. A similar problem is presented in \cite{vldb16}, please refer to
Lemma 1 in \cite{vldb16} for the proof of NP-hardness.

The exact solution of the problem is to iterate through all possible
combinations of $k$ facilities from $|F|$ facilities,
calculate the service value of each of them, and then
return the combination with the maximum value. 

Although a greedy solution exists with theoretically known best approximation ratio for the maximum coverage problem (\cite{maxcover}), the
assumption of the solution is that the objective function is 
submodular.
However, the objective function of the Max$k$CovRST problem is
{\emph non-submodular}, and thus the approximation ratio of that solution
does not hold.

\begin{lemma}
The service value function of the Max$k$CovRST problem is non-submodular.
\end{lemma}
\begin{proof}
Let $g(\cdot)$ be a function that maps a subset of a finite ground
set to a non-negative real number.
The function $g(\cdot)$ is submodular if it satisfies the natural
``diminishing returns'' property: the marginal gain from adding an
element $x$ to a set $A$ is at least as high as the marginal gain from
adding $x$ to a superset of $A$.
Formally, for all elements $x$ and all pairs of sets $A
\subseteq B$, a submodular function satisfies $g(A\cup{x}) - g(A) \ge
g(B\cup{x}) - g(B)$.

We will prove this lemma by contradiction.
Assume that the service function $\var{SO}(\cdot)$ of the
Max$k$CovRST problem is submodular.
Let $A$ be a set of facilities, and $\var{SO}(U,A)$ be the maximum
number of user trajectories that are combinedly served by $A$.
Now suppose that we add another facility $x$ to $A$ such that $\var{SO}(U,A\cup{x}) = \var{SO}(U,A)$ (i.e., no additional user is served by adding $x$).
If $\var{SO}(\cdot)$ is submodular, then $\var{SO}(U,B) \ge
\var{SO}(U,B\cup{x})$ must be true.

If we can find an instance where $\var{SO}(U,B) \not\ge
\var{SO}(U,B\cup{x})$
when $\var{SO}(U,A\cup{x}) = \var{SO}(U,A)$, $\var{SO}(\cdot)$ is 
non-submodular by contradiction.
Consider Scenario 1 where a user $u$ is served by a facility when both the
source and destination of $u$ is within $\psi$ distance from any point of the facility.
Let the source of a user $u$ be within $\psi$ from a
facility in $B$ but not $A$, and the destination of $u$ is not
within $\psi$ from either $A$ or $B$ (i.e., $u$ is not served by
either).
Let the facility $x$ be within $\psi$ distance from only the
destination of $u$.
Therefore, $u$ will be served by $B\cup{x}$ (source is served by $B$, and destination is served by $x$).
That is, $\var{SO}(U,B\cup{x}) \ge \var{SO}(U,B)$.
However $u$ is not served by $A\cup{x}$ as the source of $u$ is not
served by $A$ or $x$, i.e., $\var{SO}(U,A\cup{x}) = \var{SO}(U,A)$, which
is a contradiction.
So, $\var{SO}(\cdot)$ of the Max$k$CovRST problem is 
non-submodular.
\end{proof}

To the best of our knowledge there is no greedy solution with a
guaranteed approximation ratio for non-submodular functions for this
problem.
There are several optimization approaches, including genetic
algorithms, simulated annealing, or ant colony optimization that
could be used to find the maximum value of the objective function.
However, all of these solutions are offline and may require many
iterations to converge to an optima, so these solutions are not
suitable for the online computation of ad-hoc route planning
problems.
Therefore, we present a greedy solution of the Max$k$CovRST problem,
where the challenge is to efficiently find the users and the user segments that can be combinedly served by multiple facilities,
and compute the combined service value, as a user can be served by
multiple facilities and there can be overlaps in the service.
We exploit the \tqtree for our solution, as this structure enables us
to efficiently address these challenges.

\subsection{Greedy Solution} \label{greedy_solution}
Inspired by the greedy algorithm of Fiege~{\cite{maxcover}}, which is
the best-possible polynomial time approximation algorithm for the
maximum coverage problem, we present a greedy solution for the
Max$k$CovRST problem. A straightforward adaptation is to first compute the service value for each facility and iteratively choose a facility that serves the maximum number of users that have not been served, considering the service overlap of multiple facilities for a user. Since this straightforward approach requires to evaluate the services for all facilities and keeping track of all users who have been served by each facility, this approach can be expensive when the number of users and facilities are large. 

To overcome the above limitations, we propose a two-step greedy approach, where in the first step we compute a subset ($k^{\prime} \ge k$) of the highest serving facilities using our $k$MaxRRST algorithm. In the second step we apply the above mentioned greedy algorithm to iteratively choose a facility from those facilities that serve the maximum number of users that have not been served. We have found that this approach is highly effective in practical scenarios and can respond to queries in milliseconds. Due to space constraint, we have omitted the details, but present its experimental evaluation in Section~\ref{experiment}.

\section{Experimental Evaluation}
\label{experiment}

In this section we present the experimental evaluation for our
solutions to answer the $k$MaxRRST and Max$k$CovRST queries.
As there is no prior work that directly answers these problems, we
compare our solutions with a baseline. 

Specifically, for the $k$MaxRRST query, we compare the following three
methods: (i) Baseline (\bl): In this approach, for each facility, the
user trajectories that are within $\psi$ distance are retrieved by
executing a range query in a traditional index (in our experiments, a
quadtree).
The service value of each facility is computed, and the top-$k$
facilities are returned as the result.
(ii) TQ-tree Basic (\tqb): In this method, we use a simple TQ-tree that
hierarchically organizes user trajectories using a quadtree, but
keeps a linear list for storing trajectories in each q-node as the index
structure.
The algorithm presented in Section~\ref{algorithm} is applied on this
index.

(iii) TQ-tree Z-order (\tqz): We use our proposed TQ-tree, where the
trajectories in the hierarchical structure are ordered using a
z-curve and indexed using their z-ids in each q-node of the quadtree, and apply the algorithm presented in Section~\ref{algorithm}.
\emph{We present our approach with both TQ(B) and TQ(Z) to show the additional benefits of using the Z-ordered
bucketing in the index.}

\begin{table}
	\begin{smaller}
	\begin{center}
		\begin{tabular}{ |c|c|c| }
			\hline
			Name & \# Facilities & \# of stop points\\
			\hline 
			NY Bus Route & $2{,}024$ & $16{,}999$ \\
			Beijing Bus Route & $1{,}842$ & $21{,}489$ \\
			\hline
		\end{tabular}
		\caption{Facility trajectory datasets}
		\label{tab:facilityroute}
	\end{center}
	\end{smaller}
	\vspace{-5mm}
\end{table}

\begin{table}
	\begin{center}
		\begin{smaller}
			\begin{tabular}{ |c|c|c| }
				\hline
				Name & \# Trajectories & Type\\
				\hline 
				NY Taxi-trips (NYT) & $1{,}032{,}637$ & point-to-point \\
				NY Foursquare (NYF) & $212{,}751$ & multipoint\\
				BJ Geolife (BJG) & $30{,}266$ & multipoint\\
				\hline
			\end{tabular}
		\end{smaller}
		\caption{User trajectory datasets}
		\label{tab:usertraj}
	\end{center}
	\vspace{-24pt}
\end{table}

For the Max$k$CovRST problem, we compare four different methods: (i) Greedy baseline (G-BL) that uses baseline service evaluation strategy in the straightforward greedy approach, (ii) Greedy TQ-tree basic (G-TQ(B)) that runs our greedy solution using TQ-tree basic, (iii) Greedy TQ-tree Z-order (G-TQ(Z)) using TQ-tree Z-order, and (iv) Genetic-TQ-tree Z-order (Gn-TQ(Z)) that employs genetic algorithm using TQ-tree Z-order. 

\subsection{Experimental settings}
\label{setup}
We use Java to implement our algorithms.
All the experiments were conducted in a PC equipped with Intel core
i5-3570K processor and 8 GB of RAM.
In all of our experiments, we use in-memory data structures.
Without loss of generality our data structures can be applied for
disk-based system.

\myparagraph{Facility Datasets} We use two real bus network datasets: (i)
New York (NY) and Beijing (BJ) bus routes as our facility datasets.
Table~\ref{tab:facilityroute} shows the summary of the facility datasets.

\myparagraph{User Trajectory Datasets} To accommodate a wide range of real-world user movements with different types and volumes, we use the following three datasets: (i) Yellow taxi
trips\footnote{www.nyc.gov/html/tlc/html/about/trip\_record\_data.shtml}
in New York (NYT), (ii) Foursquare check-ins\footnote{www.kaggle.com/chetanism/foursquare-nyc-and-tokyo-checkin-dataset} in New York (NYF), and (iii) Geolife GPS traces\footnote{www.microsoft.com/en-us/download/details.aspx?id=52367} in Beijing (BJG).

The taxi-trips are essentially pairs of pick-up and drop-off
locations of passengers, and thus can be considered as user
trajectories with two points.
In contrast, the Foursquare dataset consists of user check-in data for 
different users in NY, where each check-in is a stop point for a
trajectory.
We refer these trajectories as \emph{multi-point}.
We also use Geolife GPS trajectories that contain the user movement
traces of $182$ users over three years period of time resulting $30,266$ trajectories in Beijing.
This Geolife data can also be considered as \emph{multi-point} user
trajectories.
Table~\ref{tab:usertraj} summarizes the datasets used.

\myparagraph{Performance Evaluation and Parameterization}
We studied the efficiency, scalability, and effectiveness for the
baseline and our proposed approaches by varying several
parameters.
The list of parameters with their ranges and default values in bold are shown in Table~{\ref{tab:param}}. 
For all experiments, a single parameter is varied while keeping the
rest as the default settings.

For efficiency and scalability, we studied the impact of each
parameter on (i) the runtime to calculate the service value
of a facility, and (ii) the total runtime of answering the $k$MaxRRST
query.
We evaluate the performance on the user trajectory dataset with both 
source-destination points, and multiple points.
In each case we generate $100$ sets of queries with the same settings and
report the average performance.
As the greedy solutions provide an approximate
result, we also report the effectiveness of our solutions as (i) the total number of users served, and (ii) approximation ratio.

\begin{table}
  	\begin{smaller}
		\begin{center}
			\begin{tabular}{ |l|l| }
				\hline
				Parameters & Ranges \\
				\hline 
				Routes & {\bf NY}, BJ \\
				Datasets & {\bf NYT}, NYF, BJG \\
				\# Trajectories & $203308$, {\bf 357139}, $697796$, $1032637$ \\
				\# Stops ($S$) & 8, 16, {\bf 32}, 64, 128, 256, 512\\
				\# Facilities ($N$) & 8, 16, 32, {\bf 64}, 128, 256, 512\\
				$k$ & 4, {\bf 8}, 16, 32 \\
				\hline
			\end{tabular}
			\caption{Parameters}
			\label{tab:param}
		\end{center}
		\vspace{-26pt}
	\end{smaller}
\end{table}

\begin{figure}
\centering
\subfloat[]{\includegraphics[trim = 20mm 70mm 25mm 70mm, clip,width=0.29\textwidth]{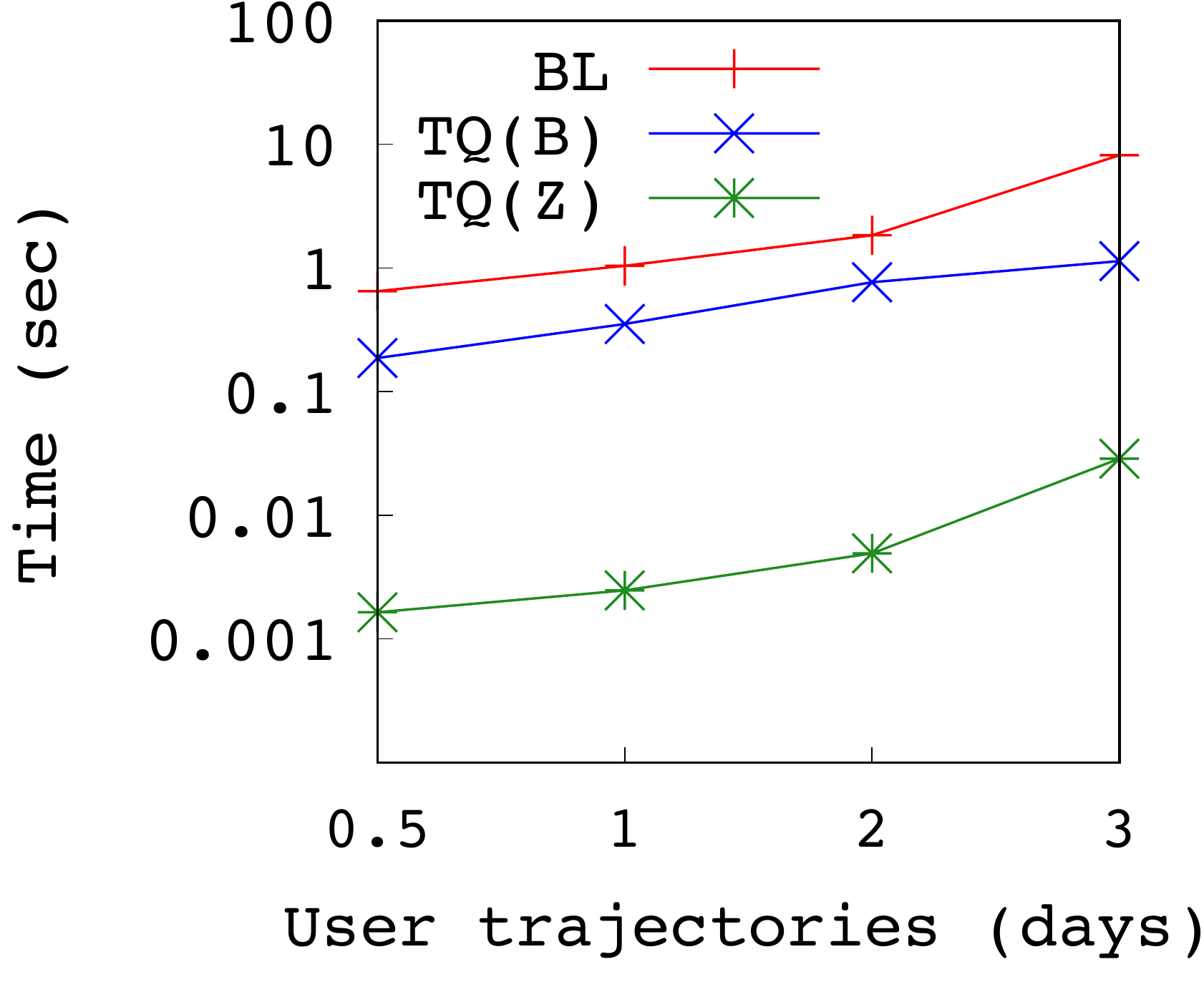}\label{fig:ow}}
\subfloat[]{\includegraphics[trim = 70mm 70mm 25mm 70mm, clip,width=0.205\textwidth]{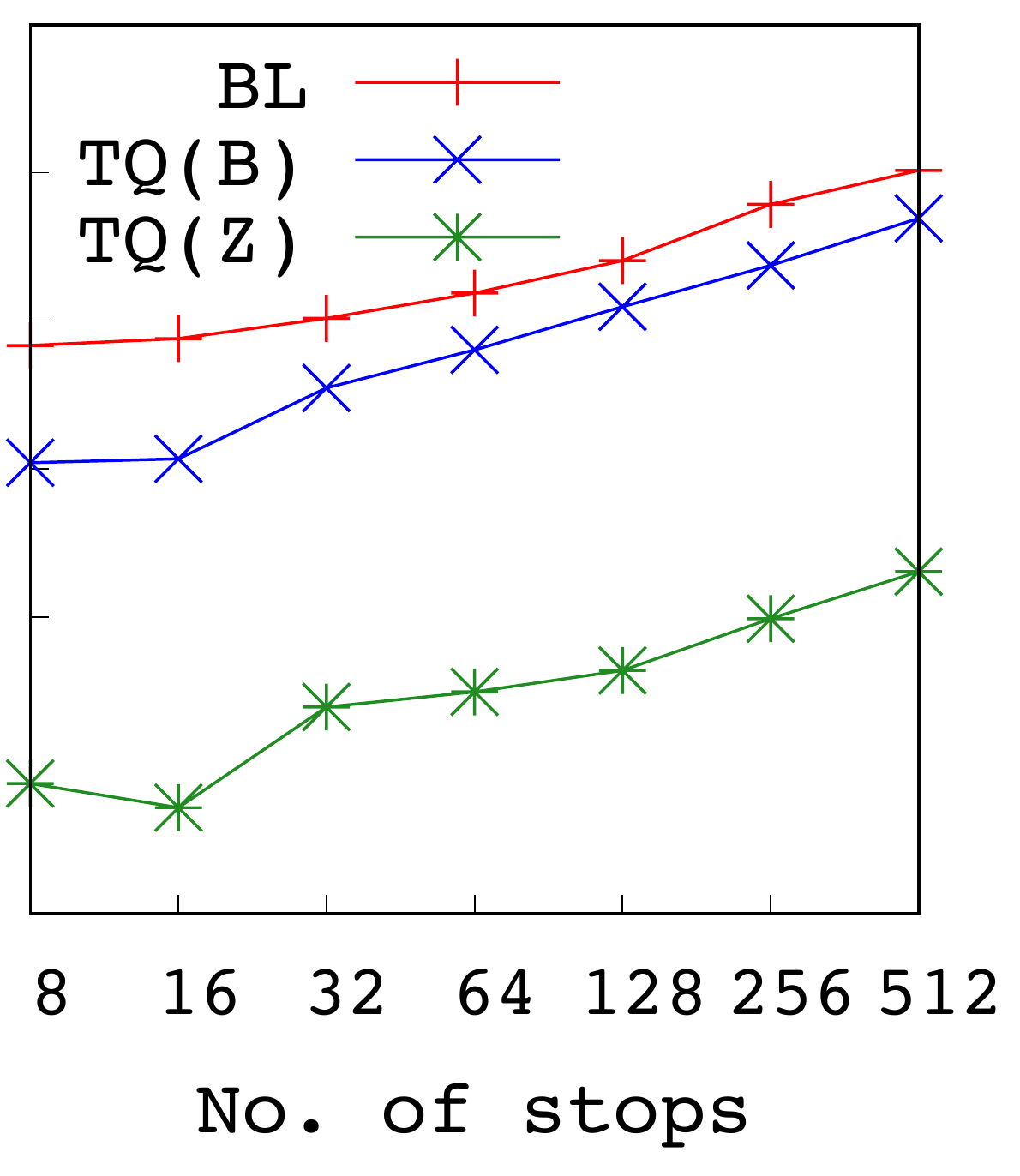}\label{fig:tw}}
\vspace{-6pt}
\caption{Evaluating service values for varying number of (a) user trajectories (b) stops in NYT dataset.}
\vspace{-22pt}
\label{fig:service}
\end{figure}

\subsection{Experimental results}

\myparagraph{1) Computing service value} 
We vary different parameters and present the processing time for calculating the service value of
a single facility in the following.

\noindent{\bf \emph{(i) No.\ of user trajectories: }} We vary the number of taxi trips in NYT dataset as $203{,}308$ (NYT-0.5), $357{,}139$ (NYT-1), $697{,}796$ (NYT-2), and $1{,}032{,}637$ (NYT-2), which corresponds to the taxi trips in $12$ hours, $1$ day, $2$ days, and $3$ days, respectively. Figure~\ref{fig:service} (a) shows the average processing time for the baseline (\bl), TQ-tree basic (\tqb), and TQ-tree Z-order (\tqz). As the \tqb organizes the trajectory segments in hierarchy in contrast to indexing points in a quadtree in \bl, \tqb is $1$ order of magnitude faster than the baseline. The spatial Z-ordering of the trajectories in \tqz results into $2$ orders of magnitude faster processing time than \tqb for calculating the service value of a single facility.

\noindent{\bf \emph{(ii) No.\ of facility stops: }}We vary the number of stops of each facility from $8$ to $512$,
and report the average processing time to compute the service value of a facility. The results (Figure~\ref{fig:service} (b)) show that \tqb and \tqz outperform the baseline by around $1$ order of magnitude and $2-3$ order of magnitude, respectively. Here, the runtime of all of the approaches gradually increase with the number of stops, as more users become eligible to be served. The benefit of the divide-and-conquer approach in the TQ-tree based approaches is higher for a lower number of stops. 

\noindent{\bf \emph{(iii) Distance threshold $\psi$: }}Although more users are likely to be eligible to be served with the increase of $\psi$, we do not observe any significant change in the performance of our proposed algorithms other than the baseline. We omit the performance graph for varying $\psi$ for brevity. 

\begin{figure*}
\centering
\subfloat[]{\includegraphics[trim = 20mm 70mm 25mm 70mm, clip,width=0.305\textwidth]{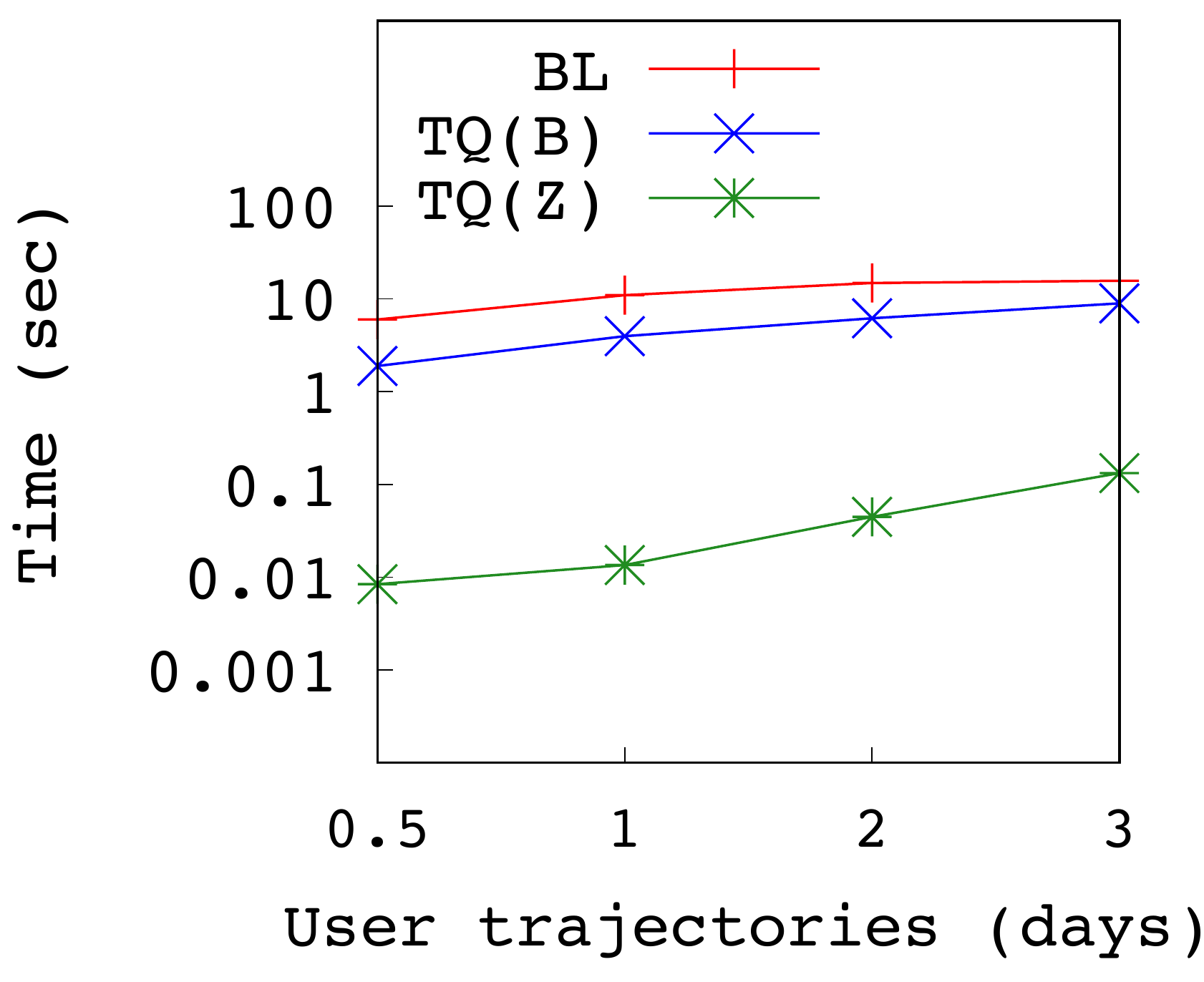}\label{fig:ow}} \hfill
\subfloat[]{\includegraphics[trim = 70mm 70mm 25mm 70mm, clip,width=0.215\textwidth]{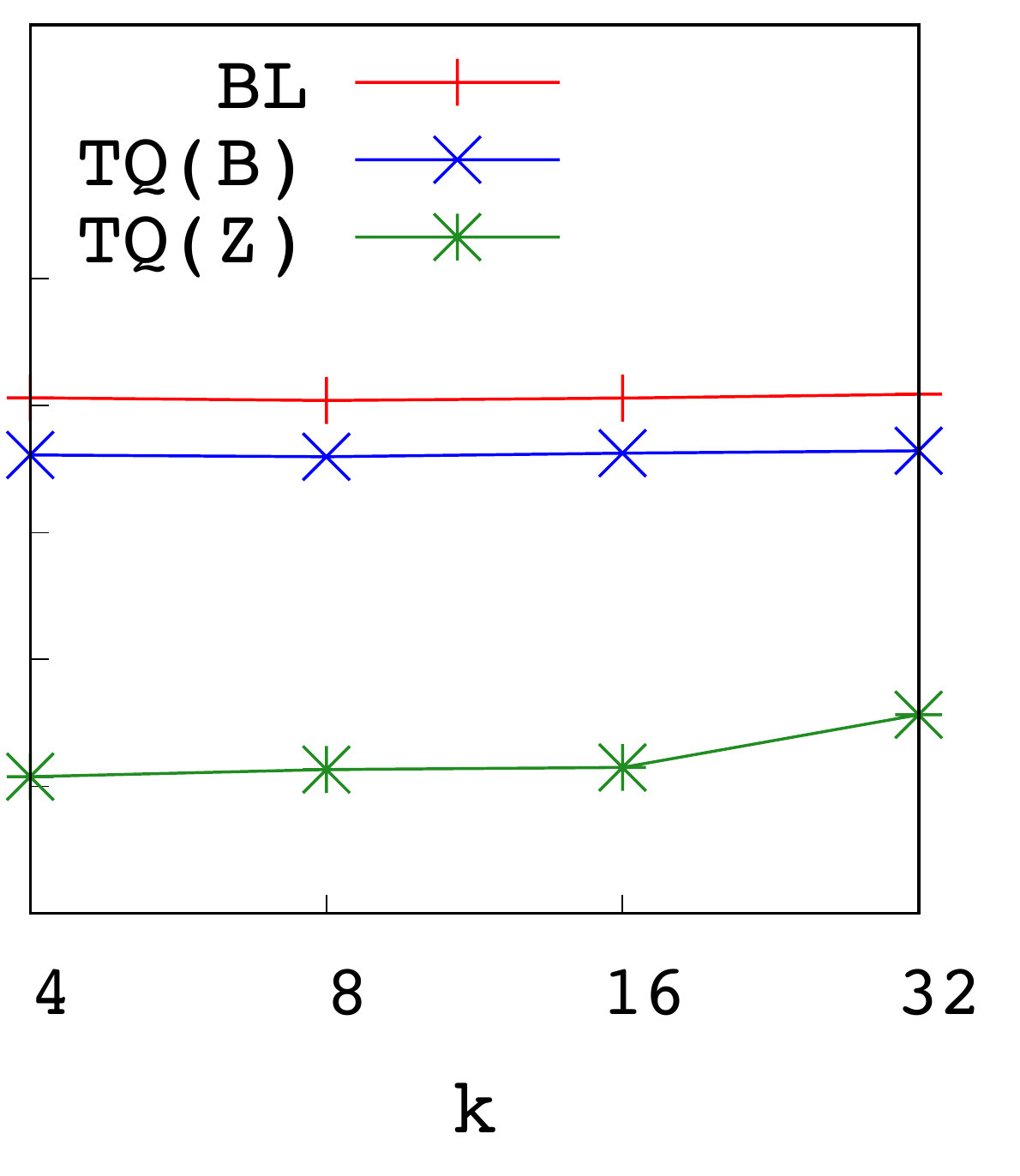}\label{fig:tw}} \hfill
\subfloat[]{\includegraphics[trim = 70mm 70mm 25mm 70mm, clip,width=0.215\textwidth]{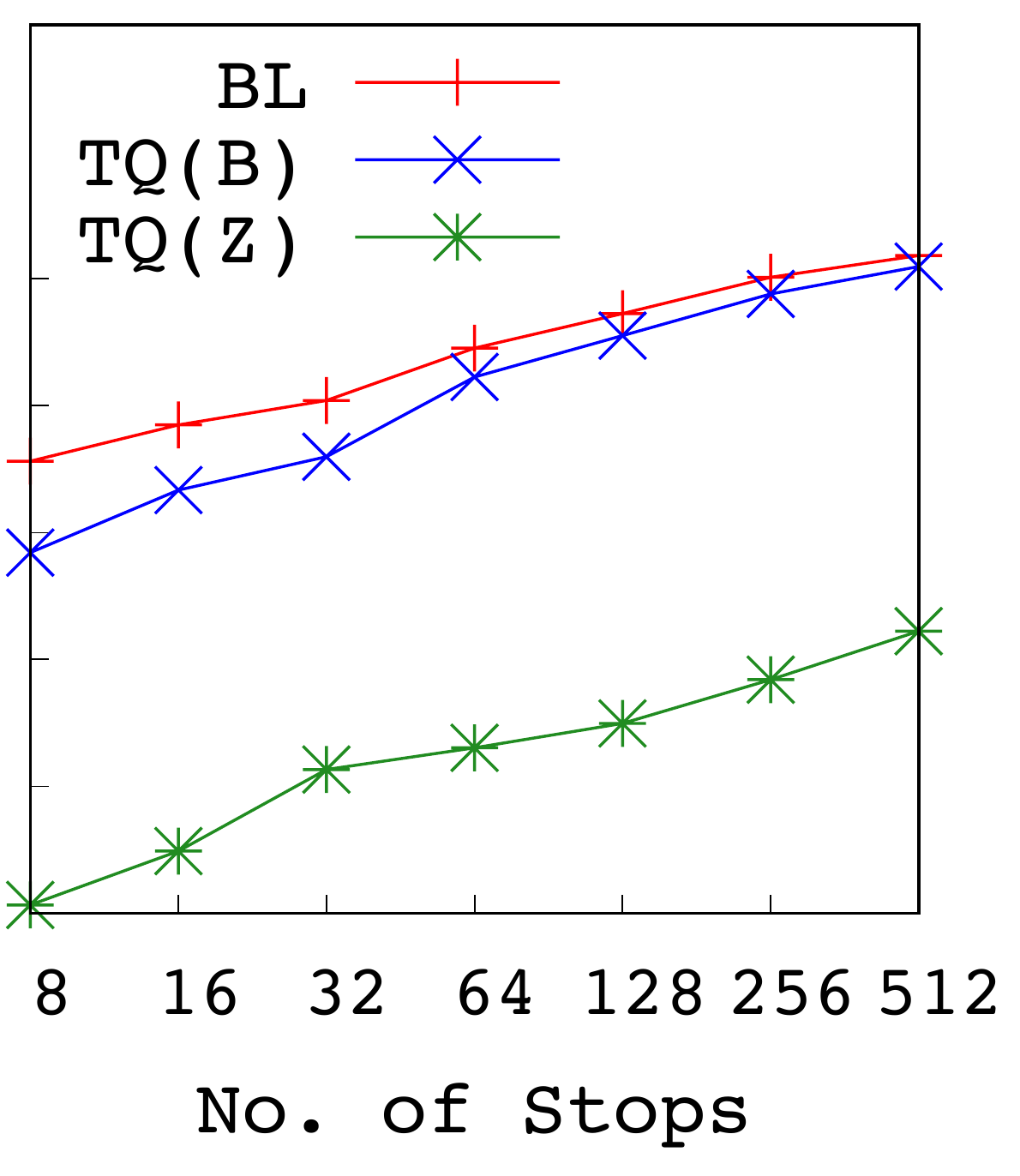}\label{fig:ow}} \hfill
\subfloat[]{\includegraphics[trim = 70mm 70mm 25mm 70mm, clip,width=0.215\textwidth]{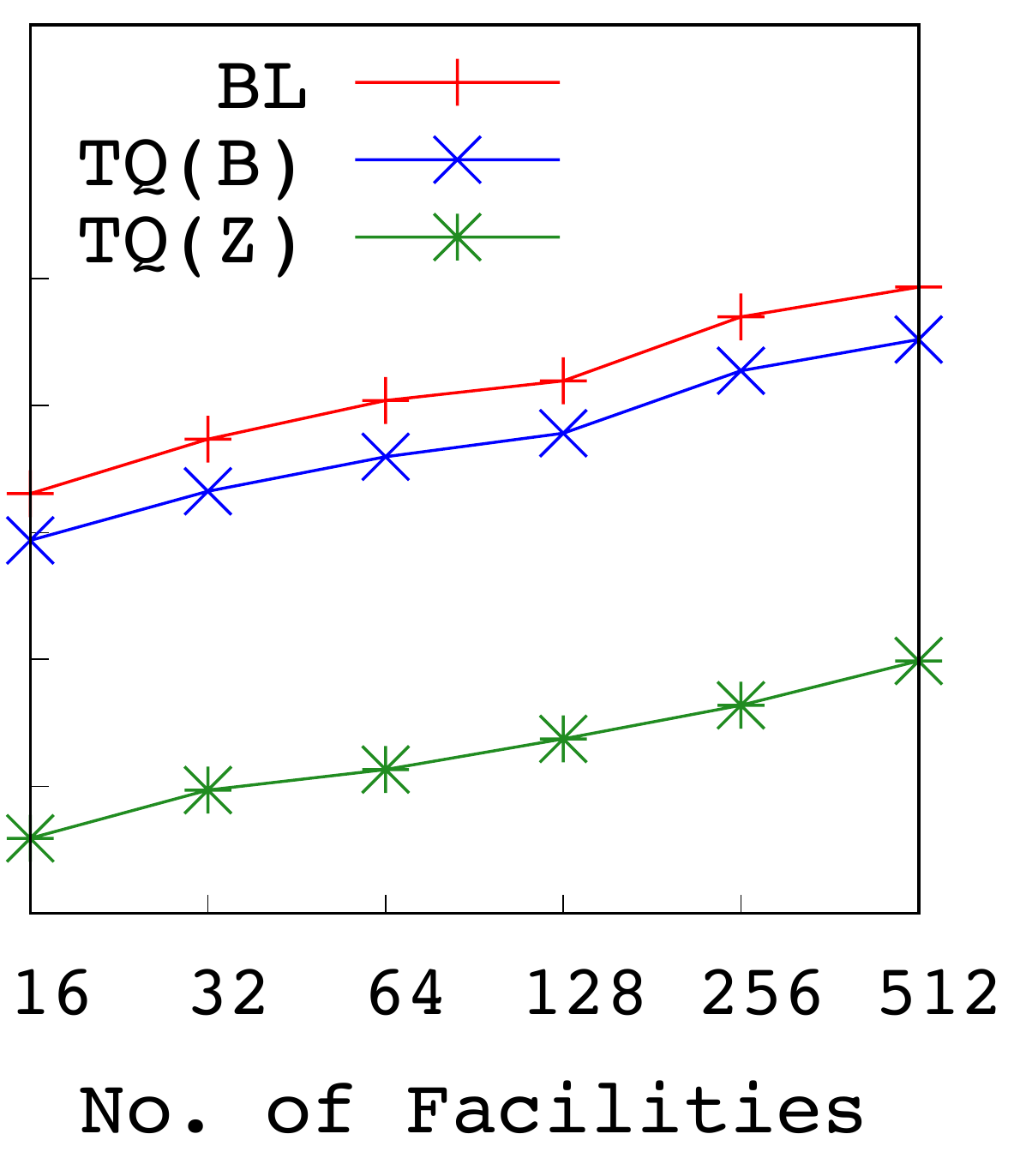}\label{fig:tw}}
\vspace{-6pt}
\caption{Evaluating $k$MaxRRST for varying (a) Users, (b) $k$, (c) stops, and (d) facilities for NYT datasets.}
\vspace{-22pt}
\label{fig:topk}
\end{figure*}

\myparagraph{2) Processing $k$MaxRRST}
We evaluate our proposed algorithms to answer $k$MaxRRST, and compare the performance with the baseline.

\noindent{\bf \emph{(i) No. of user trajectories: }}The algorithm using the \tqz index outperforms the baseline by
\emph{at least} $2-3$ orders of magnitude and \tqb by around $2$ orders of magnitude (Figure~\ref{fig:topk}(a)). As the number of trajectories in the user list of each q-node increases with the total number of user trajectories in the dataset, the benefits of TQ-tree based indexes decrease gradually. The number of unique z-ids in the Z-ordering, and the number of z-nodes also increase with the number of user trajectories, thus the processing time in \tqz increases at a higher rate than the other two approaches. 

\noindent{\bf \emph{(ii) No. of results ($k$): }}We vary the number of the required answers $k$ and compare the performances. As the baseline computes the service value of each facility and return $k$ facilities with the maximum values, the processing time of the baseline do not vary for $k$. The runtime of both TQ-tree based approaches slightly increase with the increase of $k$ as more iterations in the divide-and-conquer approach are likely to be required for a higher $k$ (Figure~\ref{fig:topk}(b)).

\noindent{\bf \emph{(iii) No. of stops: }} Similar to the previous results shown for computing the service value of a facility, the processing time of $k$MaxRRST for varying the number of stops of each facility gradually increase for all of the approaches (Figure~\ref{fig:topk}(c)). The runtime of \tqb is around $1$ order of magnitude faster than the baseline for smaller number of stops, but the benefit decreases for a higher number of stops. The reason is that the number of iterations in the divide-and-conquer approach increases with the number of stops, and the list of trajectories in the user list of a q-node needs to be searched linearly in the \tqb each time (as there is no ordering of the trajectories in the list). \tqz consistently outperforms the baseline by around $3$ orders of magnitude with the help of the efficient two-level index.

\noindent{\bf \emph{(iv) No. of facilities: }}As more computations are required to find the top-$k$ facilities from a higher number of candidate facilities, the runtime increases for each approaches at around the same rate as shown in Figure~\ref{fig:topk}(d). Although \tqb consistently outperforms the baseline, the runtime of the baseline and \tqb may not suitable for an efficient ad-hoc route planning with a higher number of facilities. The \tqz answers the query in the scale of milliseconds, and is around $3$ orders of magnitude faster than the baseline.

\begin{figure}
\centering
\subfloat[]{\includegraphics[trim = 20mm 70mm 25mm 70mm, clip,width=0.28\textwidth]{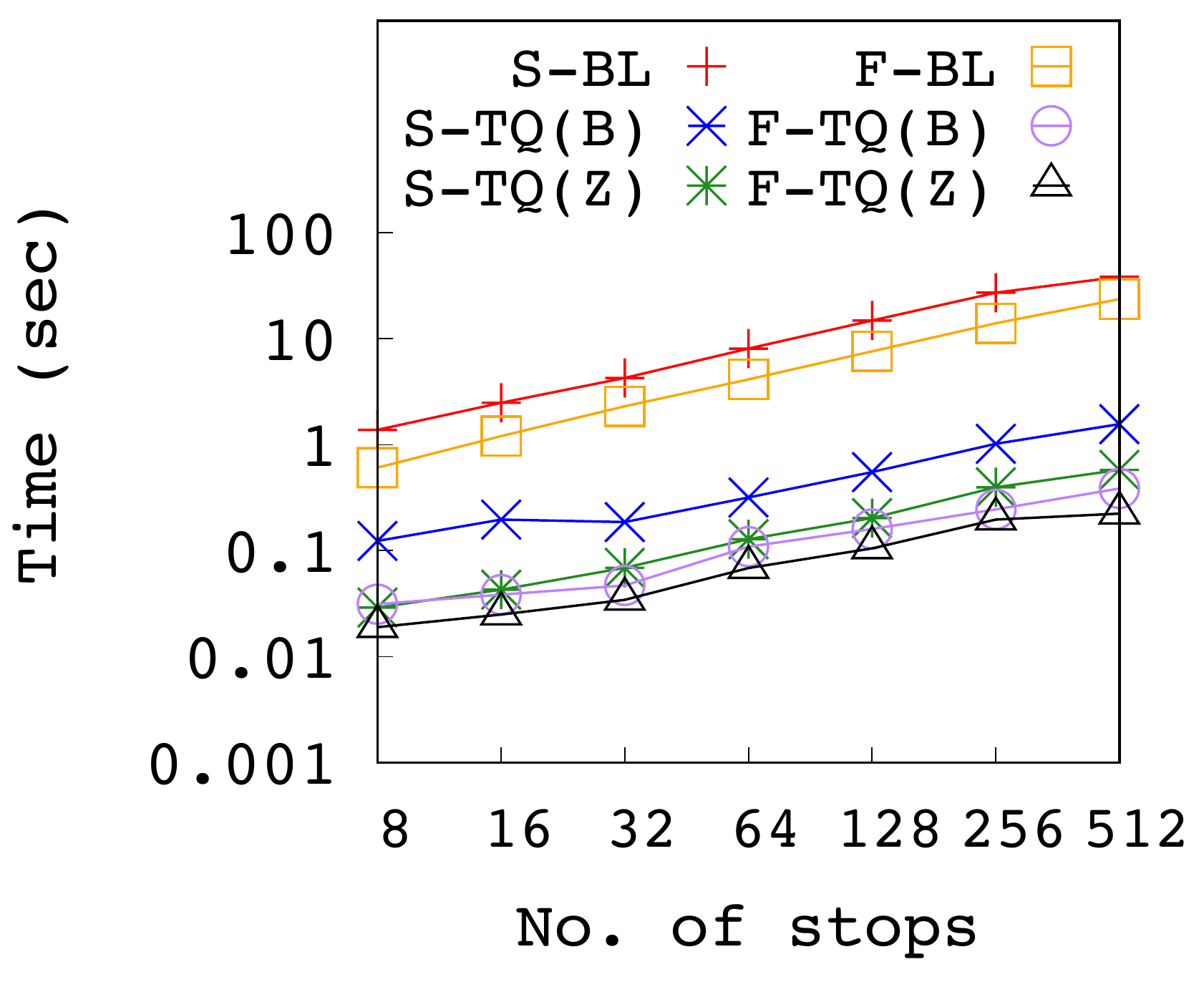}\label{fig:ow}} \hfill
\subfloat[]{\includegraphics[trim = 70mm 70mm 25mm 70mm, clip,width=0.195\textwidth]{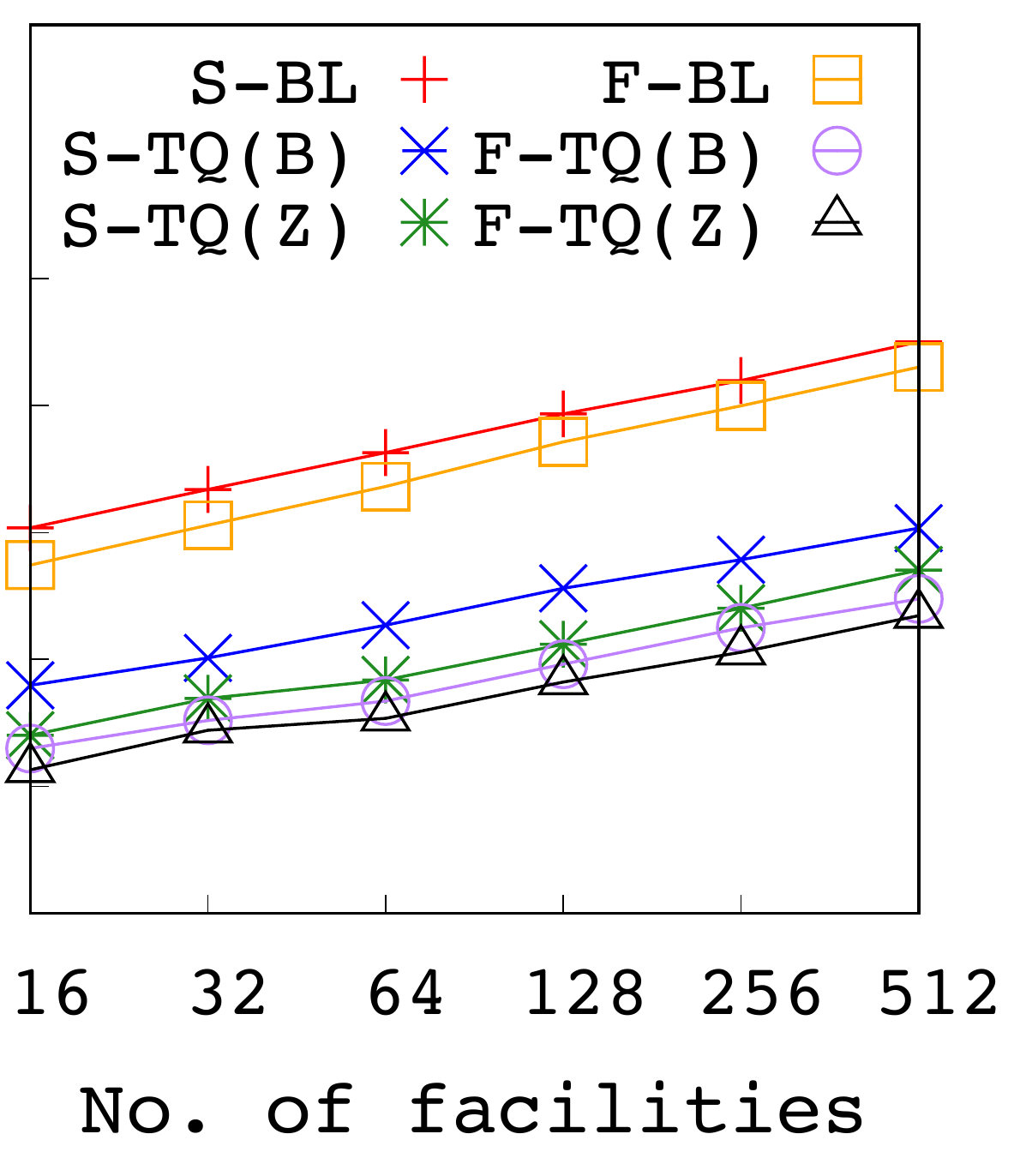}\label{fig:tw}}
\vspace{-6pt}
\caption{Evaluating $k$MaxRRST for varying number of (a) stops (b) facilities for New York Foursquare multipoint datasets.}
\vspace{-22pt}
\label{fig:topkNYF}
\end{figure}

\myparagraph{3) $k$MaxRRST for multipoint datasets}

\emph{NYF Dataset:} Since each user trajectory in the NY Foursquare-checkins dataset is a sequence of
points, we evaluated $k$MaxRRST queries using the two genaralized versions
of the index: a Segmented TQ-tree (S-TQ) and Full trajectory
TQ-tree (F-TQ) (please see Section~\ref{sec:index_gen}).
In the S-TQ version, two consecutive check-ins of a user are
considered as a segment, and all such segments of all users are
indexed using the TQ-tree. For F-TQ, we consider the sequence of checkins in a day of a user as a single multipoint trajectory, and index these
trajectories using the TQ-tree. For both approaches, we compare the performance for both the TQ-tree basic and the TQ-tree (Z-order) indexes.

Figure~\ref{fig:topkNYF} shows the results of our approaches when varying (a)
the number of stops and (b) number of facilities.
The F-TQ based approaches perform better than S-TQ as the number of trajectories increases
significantly in the segmented approach. The performance gap between the S-TQ-tree basic (S-TQ(B)) and the S-TQ(Z) is around $1$ order of magnitude, which is smaller than the previous experiments.
The underlying reason is that for smaller segments, TQ-tree contains
fewer trajectories in internal nodes of the TQ-tree, and
thus z-order based performance gain cannot be achieved.
For the same reason, we have found that the F-TQ based
approaches outperform the S-TQ based approach. In all cases, our proposed approaches for
processing $k$MaxRRST using multipoint trajectories significantly
outperform the baseline.

\begin{figure}
\centering
\subfloat[]{\includegraphics[trim = 20mm 70mm 25mm 70mm, clip,width=0.28\textwidth]{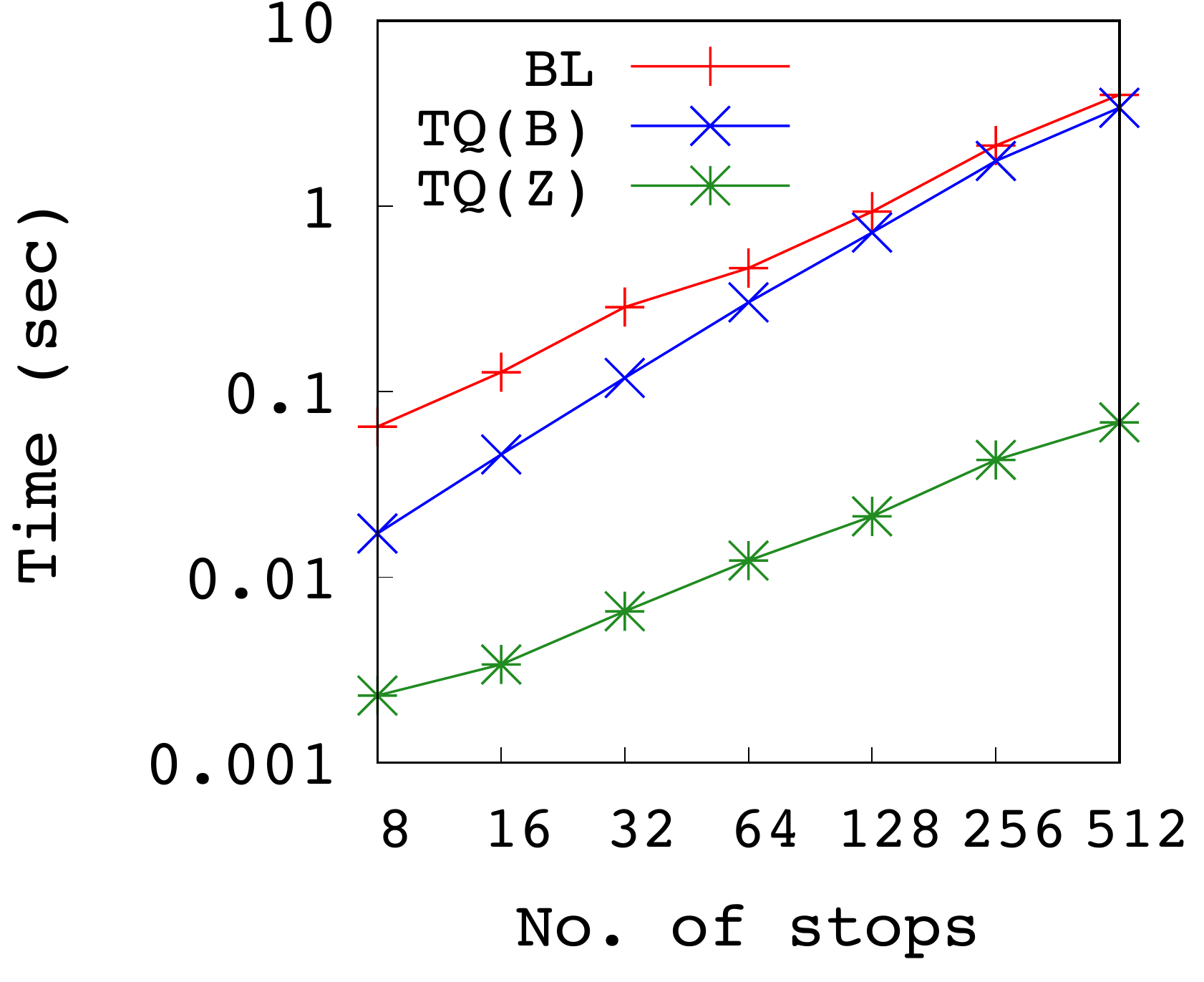}\label{fig:ow}} \hfill
\subfloat[]{\includegraphics[trim = 70mm 70mm 25mm 70mm, clip,width=0.195\textwidth]{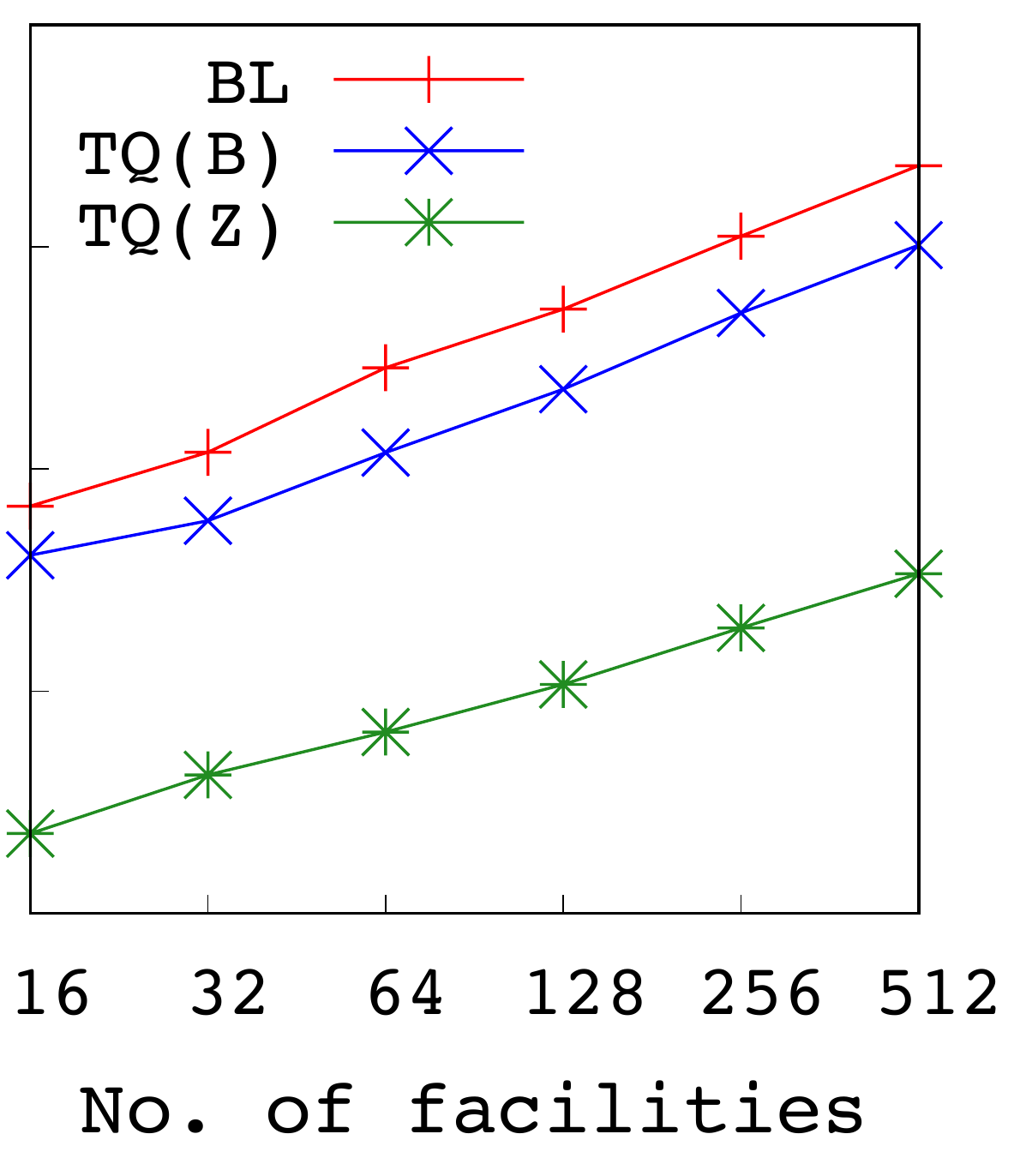}\label{fig:tw}}
\vspace{-6pt}
\caption{Evaluating $k$MaxRRST for varying number of (a) stops (b) facilities for Beijing Geolife multipoint datasets.}
\vspace{-22pt}
\label{fig:topkBJG}
\end{figure}

\emph{BJG dataset: }We evaluate our algorithms on another multipoint trajectory dataset from the Geolife project.
Since the dataset is small, we run the experiments with the segmented
TQ-tree approach, and consider every pair of points as a single
trajectory.
Figure~\ref{fig:topkBJG} shows that even for a small dataset our
TQ-tree based approaches significantly outperform the baseline.

 \begin{figure*}
\centering
\subfloat[]{\includegraphics[trim = 20mm 70mm 25mm 70mm, clip,width=0.25\textwidth]{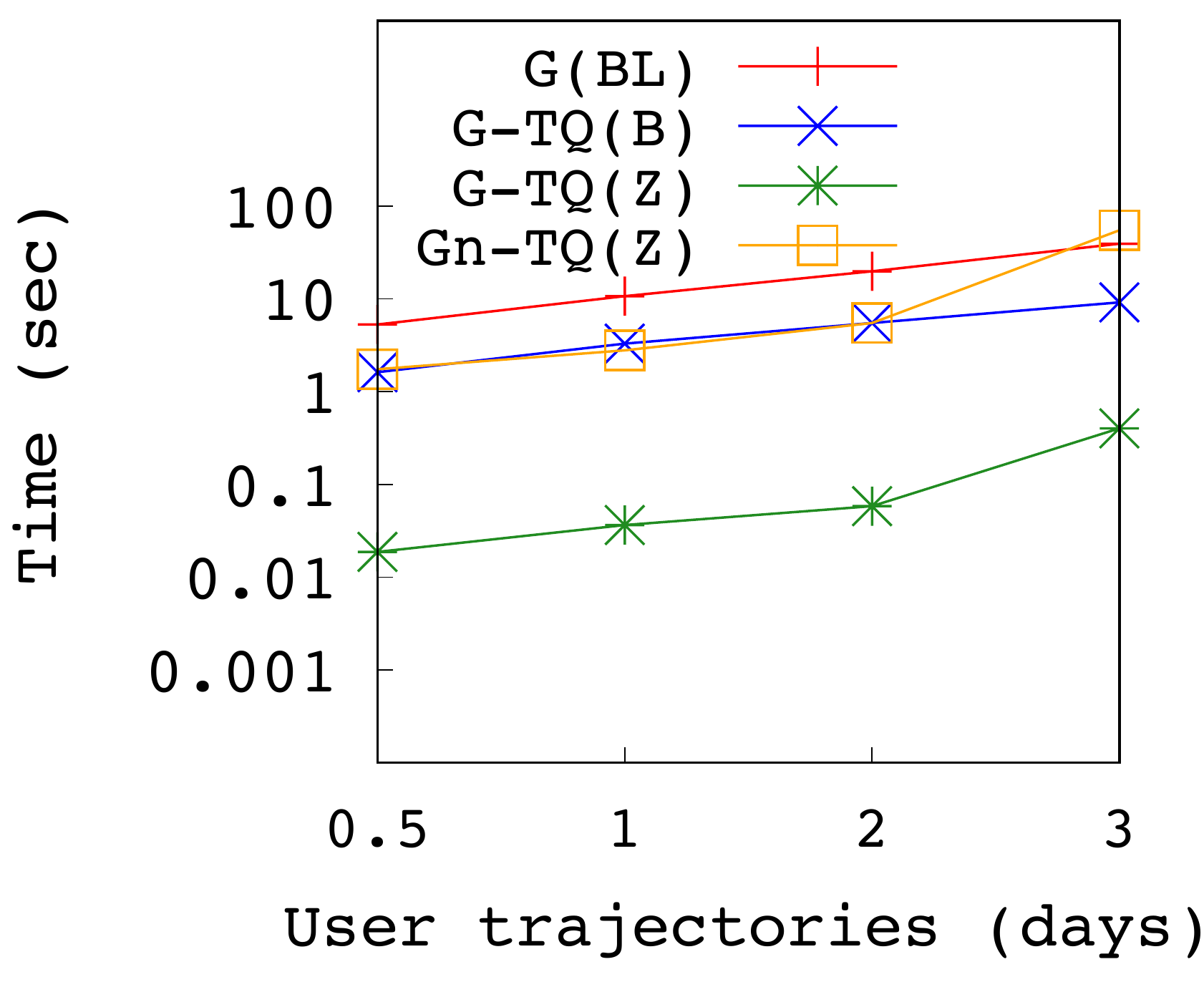}\label{fig:ow}} \hfill
\subfloat[]{\includegraphics[trim = 20mm 20mm 50mm 20mm, clip,width=0.25\textwidth]{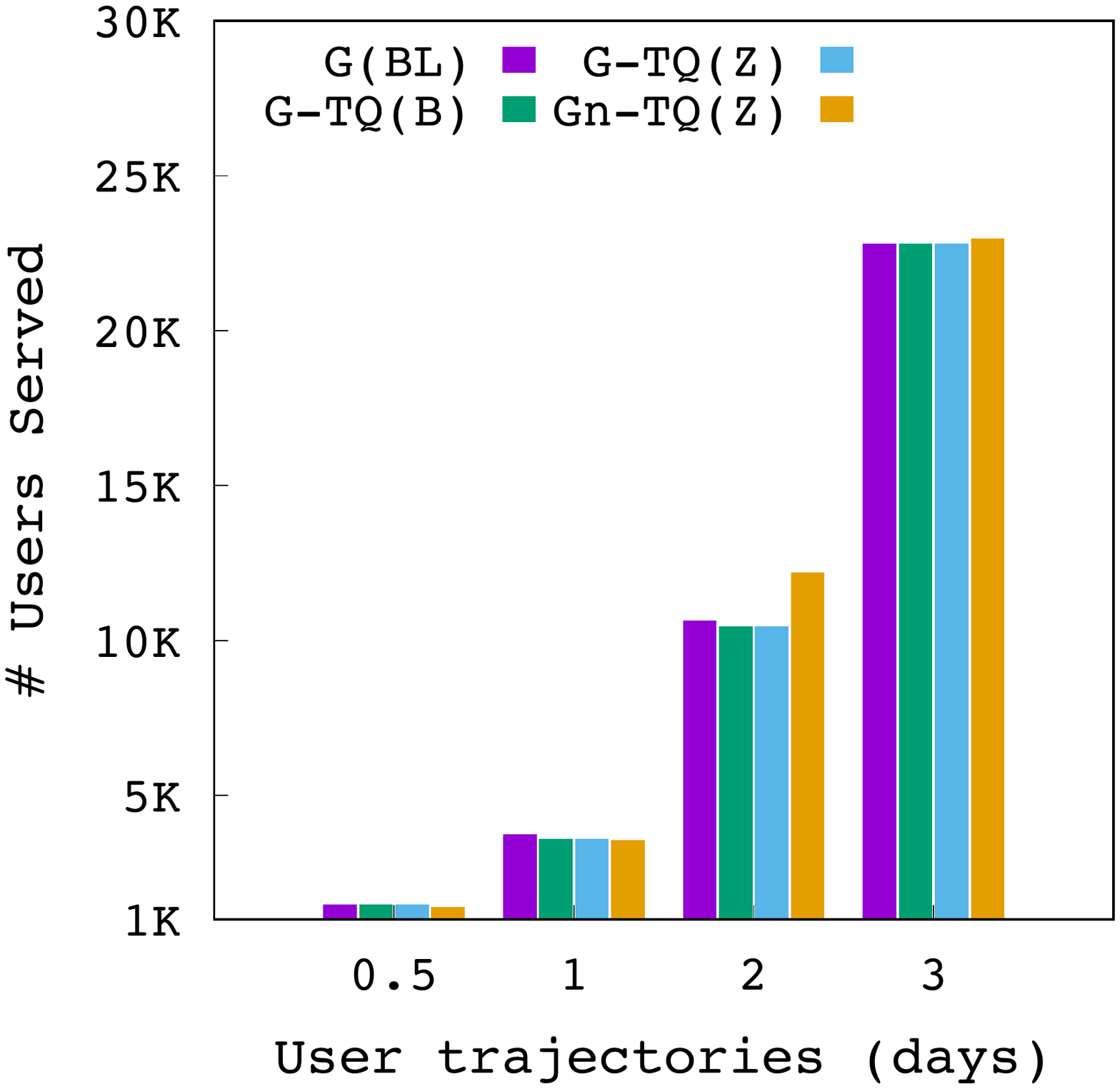}\label{fig:tw}} \hfill
\subfloat[]{\includegraphics[trim = 20mm 70mm 25mm 70mm, clip,width=0.25\textwidth]{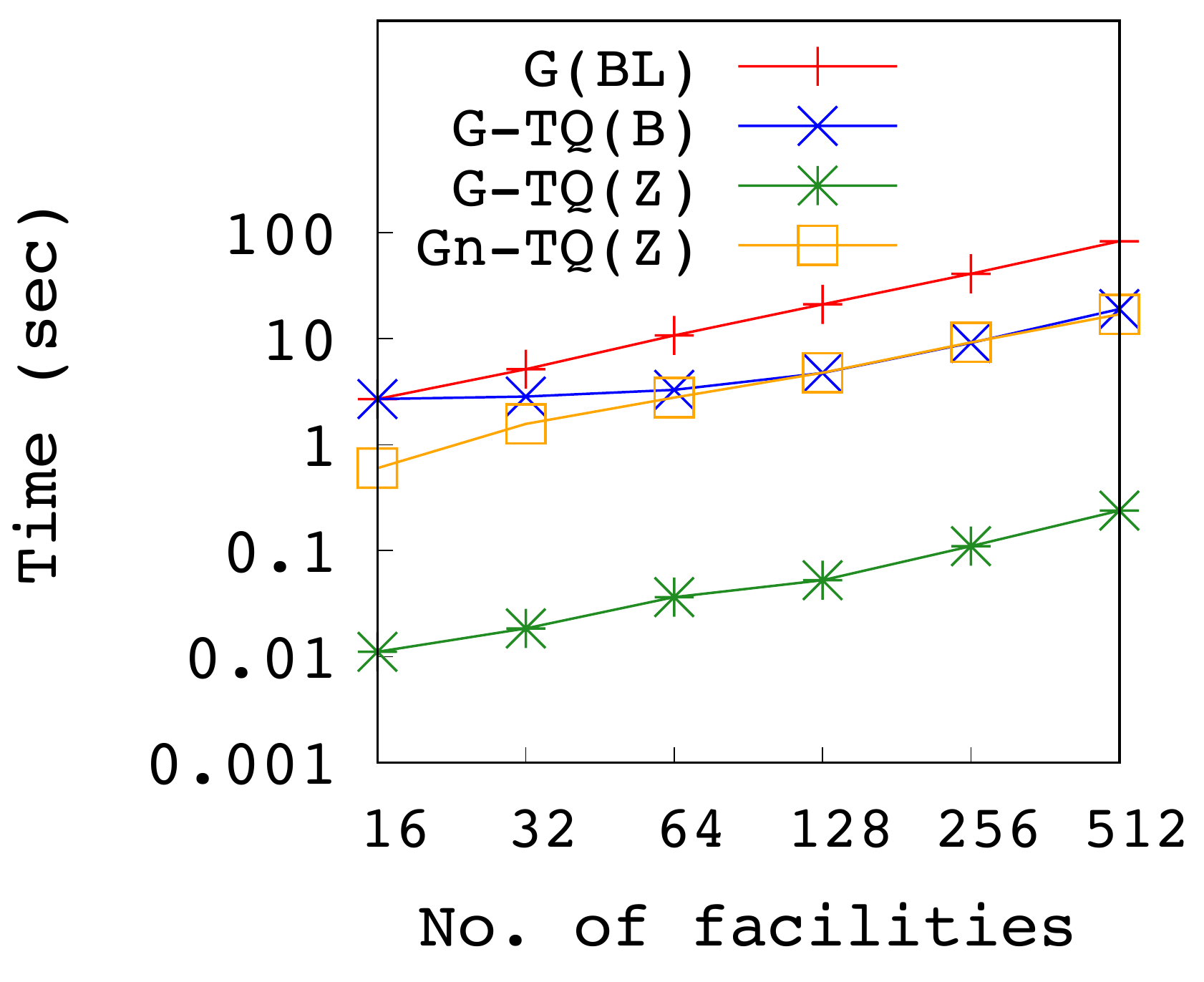}\label{fig:ow}} \hfill
\subfloat[]{\includegraphics[trim = 20mm 20mm 50mm 20mm,clip,width=0.25\textwidth]{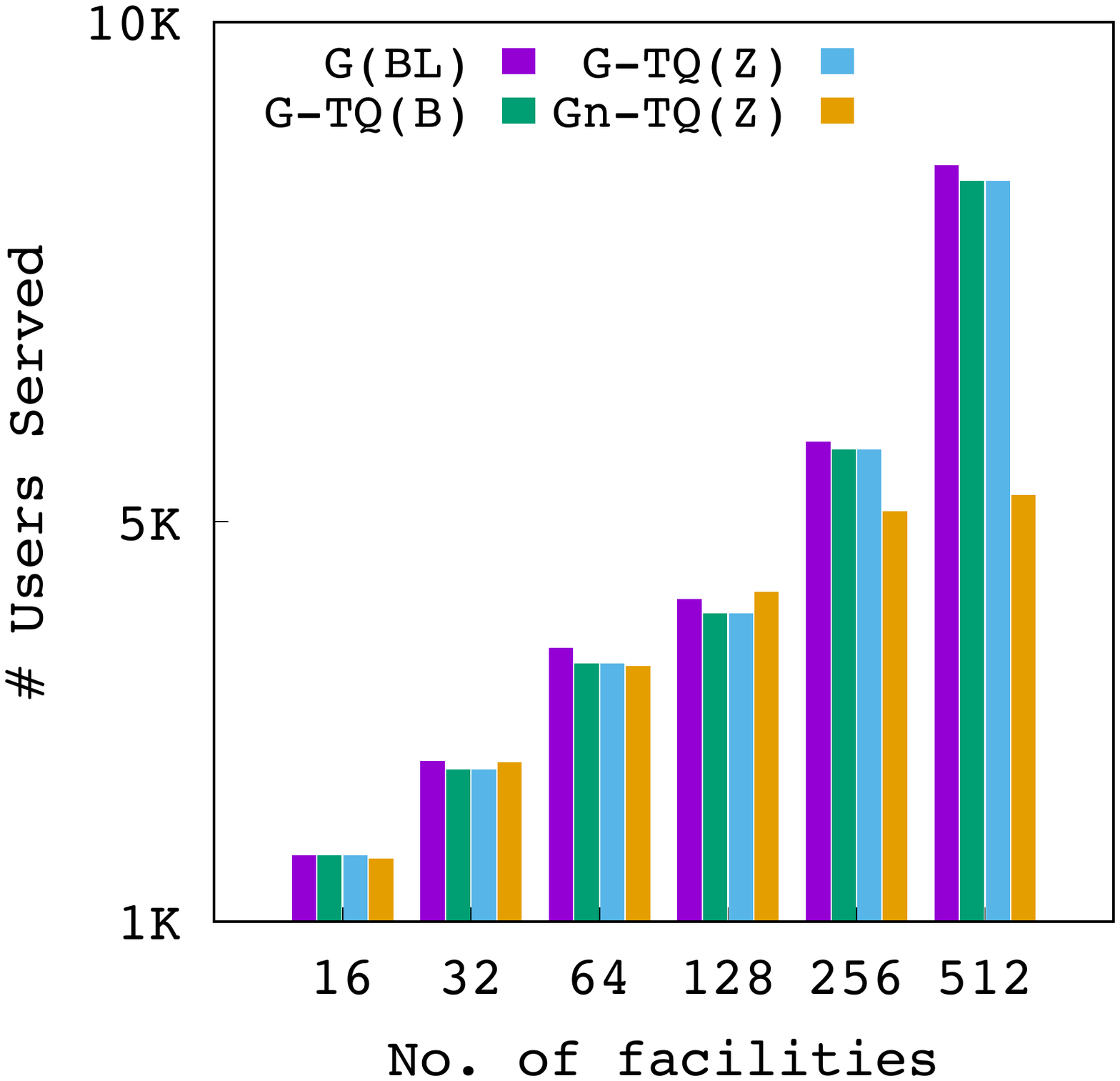}\label{fig:tw}}
\vspace{-6pt}
\caption{Evaluating Max$k$CovRST for varying (a)-(b) users (c)-(d) facilities for NYT datasets.}
\vspace{-22pt}
\label{fig:maxkNYT}
\end{figure*}

\begin{figure}
\centering
\subfloat[]{\includegraphics[trim = 45mm 20mm 45mm 25mm, clip,width=0.27\textwidth]{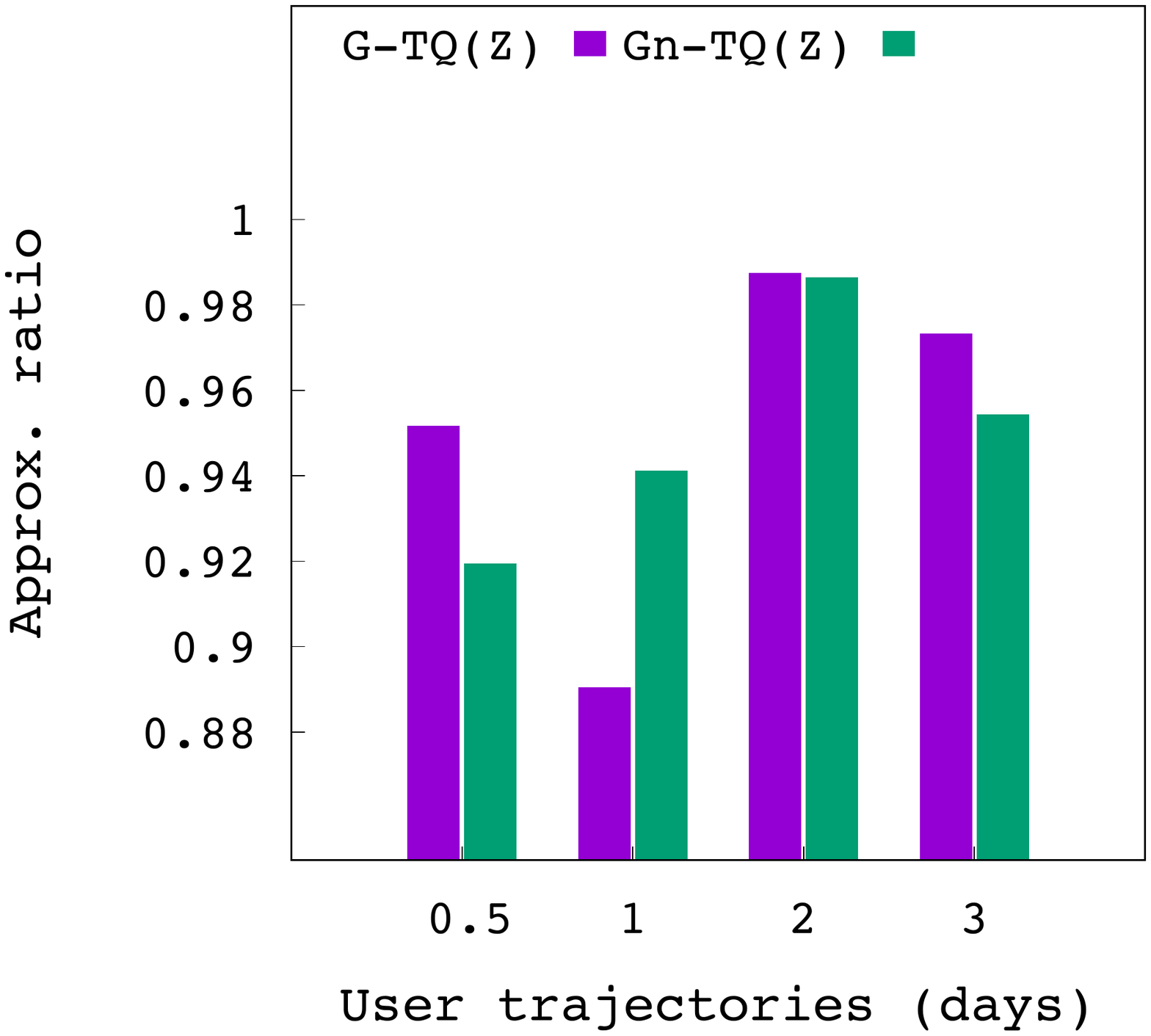}\label{fig:ow}} \hfill
\subfloat[]{\includegraphics[trim = 90mm 20mm 45mm 25mm, clip,width=0.205\textwidth]{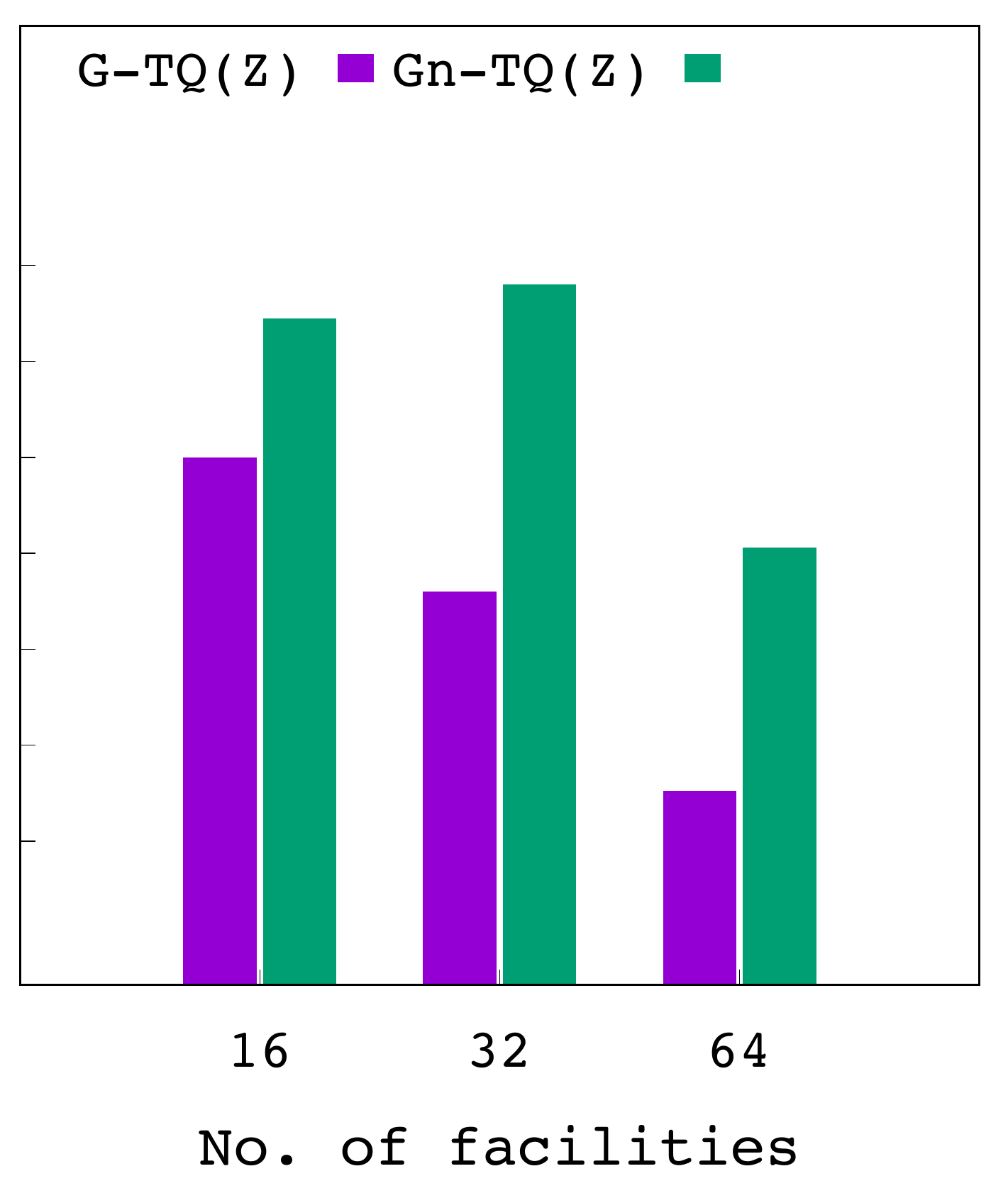}\label{fig:tw}}
\vspace{-6pt}
\caption{Approximation ratio for evaluating Max$k$CovRST for varying (a) users (b) facilities for NYT datasets.}
\vspace{-22pt}
\label{fig:ratio}
\end{figure}
 
\myparagraph{4) Evaluate Max$k$CovRST}
We also evaluated the effectiveness and efficiency of our greedy algorithm, and compare between the competitive approaches
Figure~\ref{fig:maxkNYT} shows the processing time for varying the number of users and facilities for processing Max$k$CovRST in NYT dataset. The G-TQ(Z) outperforms other approaches by a big margin.

We also evaluate the quality of our approaches in terms of number of users served (Figure~\ref{fig:maxkNYT}(b), Figure~\ref{fig:maxkNYT}(d)) and as the approximation ratio with the exact solution (Figure~\ref{fig:ratio}).  Experimental evaluation shows that the approximation ratio of our greedy TQ(Z) is close to the exact solution in most of the cases, and at least achieves $0.9$ ratio. The genetic algorithm (20 iterations) performs poorly in terms of the number of users served when the number of facilities is large (Figure~\ref{fig:maxkNYT}(d)).

\emph{\textbf{Index construction time:} } We evaluate the index construction cost of both TQ(B) and TQ(Z). The index construction for $203{,}308$, $357{,}139$, $697{,}796$, and $1{,}032{,}637$ users trips of NYT data takes only 0.74, 0.95, 2.42, 3.74 secs, respectively for TQ(B), and 1.03, 1.86, 4,23, 9.95 secs, respectively for TQ(Z). The index construction times for the other datasets are also less than a second for both indexes.
\section {Related Work} \label{sec:related}
The related body of work mostly includes
studies in trajectory indexing and query processing, facility
location selection problems, and the route planning algorithms.

\subsection{Trajectory Indexing and Queries}
There have been studies to find user trajectories, including finding human mobility patterns~{\cite{LiuLYDFXXW14}}, detecting taxi trajectories~{\cite{ChenZCLSL11}}, etc. However, as we only use the user trajectories directly as input, the methods for constructing trajectories is outside the scope of this paper.
Relevant studies on trajectories can be categorized mainly as: (i) trajectory search by similarity, (ii) trajectory search by point
locations, and (iii) reverse $k$ nearest neighbor (R$k$NN) queries on
trajectories. Studies on each of these categories propose a variety of different indexes and algorithms.
We also discuss other approaches addressing trajectory storage and retrieval in general.

\myparagraph{Trajectory Search by Similarity}
Frentzos et al.~{\cite{FrentzosGT07}} define a dissimilarity metric
between two trajectories and apply a best-first technique to return the $k$ most similar trajectories to a query trajectory.
Chen et al.~{\cite{ChenOO05}} address the problem of finding similar
trajectories based on the edit distance. A comparative review of different measures of similarity is presented in~\cite{WangSZSZ13}.
Shang et al.~{\cite{ShangDYXZK12}} study a variant of this problem,
where both location and textual attributes of the
trajectories are considered.

The significance of each point in a query trajectory is taken into
consideration in~\cite{ShangDZJKZ14}, where users can
specify a weight for each point in the query trajectory to find the $k$ most similar trajectories using the weights in the similarity function.
The general idea is to take each point along the query trajectory and
check whether a circle with the point as centre and a
threshold based on the user-defined weight as the radius, touches any
data trajectory.
Based on whether a trajectory is touched for each point or some points of the query trajectory, a lower and an upper bound of similarity is
calculated, and different pruning techniques are applied.
However, this approach is not directly amenable to our problem, as
this computation needs to be repeated for each of the
facilities, which will incur a high computational cost and
unnecessary, repeated retrieval of trajectories.
Moreover, this approach cannot efficiently answer the Max$k$CovRST query, where an user trajectory can be served jointly
by multiple facility trajectories.

\myparagraph{Trajectory Search by Point Location}
Given a set of query points, Tang et al.~\cite{TangZXYYH11} answer the $k$ nearest trajectories, where the distance to a trajectory is calculated as the sum of the distances from each
query point to its nearest point in that trajectory.
Han et al.~\cite{HanCZLW14} find the top-$k$ trajectories that are
close to the set of query points with respect to traveling time.
Given a set of query locations, finding the top-$k$ trajectories that
best connect the points (either maintaining order or unordered) are studied.
Each of these solutions use variations on the $R$-tree to store
trajectory points. As computing both the individual and partial service is important in our case, these techniques are not useful for our problem. Adapting these approaches would affect our pruning strategy greatly, resulting in higher
computational complexity as it will not be easy to exclude the 
inter-node trajectories by indexing the
points independently. Also, the queries (facilities) in our problem are also trajectories, not just points.

\myparagraph{Reverse $k$NN Trajectory Queries}
Given a set of user trajectories $U$, a set of facility (bus)
trajectories $F$, and a new facility trajectory $f \not\in F$, an R$k$NN query returns
the user trajectories from $U$ for which $f$ is one of the $k$
nearest facilities.
Wang et al.~\cite{WangBCSC17} address this problem where they
consider each user trajectory as transitions (trajectories with just
pickup and drop-off points).
In contrast to their work, we assume a user can be served by a
facility if the trajectories (stop points) are sufficiently
close. Moreover, their approach cannot be used to solve the Max$k$CovRST
query, where a user trajectory can be served jointly by multiple
trajectories.

\myparagraph{Other Storage Techniques}

Other index structures, e.g., {\em TrajTree}~\cite{RanuPTDR15}, {\em
SharkDB}~\cite{WangZZS15} are also proposed to efficiently
store trajectories.
However, as segmentation of trajectories is required to construct
TrajTree~\cite{RanuPTDR15}, this index is not amenable to 
our problem when computing the served portions of the individual
user trajectories.
SharkDB~\cite{WangZZS15} is an in-memory column oriented
timestamped storage solution used for indexing trajectory data.
This index can support $k$NN and window queries in the spatio-temporal
domain, but cannot be directly applied to solve Max$k$CovRST where identifying trajectories can be partially
served by a facility trajectory.

\subsection{Facility Location Selection Problem}
Several studies have investigated the problem of finding a location
or a region in space to establish a new facility such that the
facility can serve the maximum number of customers based on different
optimization criteria.
The min-dist selection problem finds a location for a facility
such that the average distance from each customer to the closest
facility~\cite{mindist} is minimized.
A similar problem was presented by Papadias et al.
~\cite{groupnn}, that finds a location that minimizes the sum of the
distances from the users.

A {\em Maximizing Bichromatic Reverse $k$NN} query~\cite{WongOYFL09,ZhouWLLH11} finds the optimal region in
space to place a new facility $f$ such that the number of customers
for which $f$ is one of the $k$NNs, is maximized.

These queries focus on point data or regions in space, thus they are not
directly applicable to our problem on trajectories.

\subsection {Route/Trip Planning}
Bus network design is known to be a complex, non-linear, non-convex,
multi-objective NP-hard problem~\cite{ChenZLZ14}.
Based on mobility patterns, there are a number of solutions for recommending driving route~\cite{ChenZZLAL13}, discovering popular routes~\cite{ChenSZ11},
or recommending modification of existing routes/introducing new routes~\cite{LiuLYDFXXW14}.
The MaxR$k$NNT query was proposed by Wang et al.~\cite{WangBCSC17},
which focus on constructing an optimal bus route based on a Reverse
$k$NN trajectory query.
Lyu et al.~\cite{LyuCLLZ16} propose new bus routes by processing taxi
trajectories while other works~\cite{HuangBJW14,ChenZZLAL13} aimed at
constructing bus routes by analyzing hotspots of user
trajectories.

Variants of ride-sharing problems have also received considerable
attention in literature. Ma et al.~\cite{MaZW13} present techniques to
dynamically plan taxi ride-sharing.
Given a set of location preferences,
recommending a travel trajectory that passes through those locations
have been extensively studied~\cite{ZhuXLZLZ15}.

In contrast, our proposed $k$MaxRRST and Max$k$CovRST queries focus on finding a
subset of the query trajectories that serve the highest number of
users locally and globally, respectively.
So unlike some of the aforementioned works that find the best route
offline, we can support online query processing.

\section{Conclusion}
\label{conclusion}

In this paper, we have proposed a novel index structure, the Trajectory Quadtree ($TQ$-tree) that utilizes a Quadtree to hierarchically organize trajectories into different Quadtree nodes, and then applies a z-ordering to further organize the trajectories by spatial locality inside each node. We have demonstrated that such a structure is highly effective in pruning the trajectory search space for processing a new class of coverage queries for trajectory databases:  (i) {\em Maximizing Reverse Range Search on Trajectories} (MaxRRST); and (ii) a {\em Maximum $k$ Coverage Range Search on Trajectories} (Max$k$CovRST). We have evaluated our algorithms through an extensive experimental study on several real datasets, and demonstrated that our $TQ$-tree based algorithms outperform common baselines by two to three orders of magnitude in terms of processing the coverage queries on trajectory databases. In future, we will investigate the effectiveness of the $TQ$-tree for other variants of queries on trajectory databases.

\begin{small}
\footnotesize

\bibliographystyle{IEEEtran}
\bibliography{vldb-maxtraj}

\end{small}

\end{document}